\documentclass[12pt]{article}

\usepackage[margin=1in]{geometry}  
\usepackage{graphicx}              
\usepackage{amsmath}               
\usepackage{amsfonts}              
\usepackage{amsthm}  
\usepackage{mathrsfs} 
\usepackage{mathtools}   
\usepackage[numbers,sort&compress]{natbib}            

\newcommand\abs[1]{\left|#1\right|}
\numberwithin{equation}{subsection}

\newtheorem{thm}{Theorem}[subsection]
\newtheorem{lem}[thm]{Lemma}
\newtheorem{prop}[thm]{Proposition}

\begin{document}

\nocite{*}

\title{An indefinite system of equations governing the fractional quantum Hall effect}

\author{Joseph Esposito\footnote{Email address: esposito@cims.nyu.edu}\\ 
Courant Institute of Mathematical Sciences\\
New York University\\
New York, NY 10012}

\maketitle
\providecommand{\keywords}[1]{{{Keywords:}} #1}
\providecommand{\MSC}[1]{{{MSC numbers:}} #1}
\begin{abstract}
In this paper, we establish existence of solutions to an indefinite coupled non-linear system.  We use partial coercivity to establish existence of a critical point to an indefinite functional and thus the existence of solutions on both bounded domains and doubly-periodic domains.  We then take the limit of the bounded domain to extend solutions on a bounded domain to the full-plane.  We also provide asymptotic decay analysis for solutions over the full-plane.  Necessary and sufficient conditions are also provided for the doubly periodic domain.
\end{abstract}
\\\\
\MSC{35J50, 35J60, 81V70}



\section{Introduction}


The quantum Hall effect is a physical phenomenon that is observed in condensed-matter physics.  This phenomenon occurs when a two-dimensional gas, such as that found in GaAs and AlGaAs semiconductors \cite{chakraborty2013quantum}, is subjected to a very low temperature and a magnetic field orthogonal to the plane the gas lies in.  Consequently, a quantity called the Hall conductance is quantized in the form of a filling factor, which comes from the interaction of the electrons \cite{laugh,klitz}.  This filling factor may take on integer values or fractional values and is described  by both the electron density and the quantization of magnetic flux.  When the filling factor takes on fractional values, this phenomenon is called the fractional quantum Hall effect, or, FQHE  \cite{levkivskyi2012mesoscopic,bondsling3,PhysRevB.95.125302,chakraborty2013quantum,taylor2002quantum}. The model describing the FQHE in a two-dimensional double-layered system was described with the use of Chern-Simons theory \cite{hans,hany,lieby,lin,bt,taran,wangy,yangy1,yangy2,yangy3,yang2013solitons,frol1,jackiw,Thoules1,Thoules2}.   

The goal of this paper is to study the coupled nonlinear elliptic system described in \cite{ichinose1997topological,medina}, but this time allow for the coupling matrix to be indefinite. That is, we study the system over $\mathbb{R}^2$ with $x=(x_1,x_2)$,
\begin{equation}\label{maineq}
\begin{cases}
\Delta u=4k_{11}e^u+4k_{12}e^v-4+4\pi\displaystyle\sum\limits_{j=1}^{N_1}\delta_{p_j}(x)\\
\Delta v=4k_{21}e^u+4k_{22}e^v-4+4\pi\displaystyle\sum\limits_{j=1}^{N_2}\delta_{q_j}(x)\\
u,v\rightarrow-\ln2\qquad\text{as}\qquad\abs{x}\rightarrow\infty
\end{cases}
\end{equation}
where
\begin{equation}\label{kmat}
K=(k_{ij})=\frac{1}{p}
\begin{pmatrix}
p+q&p-q\\p-q&p+q
\end{pmatrix},
\end{equation}
and the point vortices are given by the Dirac distributions centered at $p_j$ and $q_j$ in $\mathbb{R}^2$.  The parameters $p,q\neq0$ are coupling parameters and are taken to be real numbers.   Finally, the vortex numbers of the upper and lower layers are given by $N_1$ and $N_2$ respectively. 

For such a system, Ichinose and Sekiguchi discussed radially symmetric solutions for the $(m,m,n)$ Halperin state \cite{ichinose1997topological}.  Later, in \cite{medina}, Medina established doubly-periodic and full-plane solutions when the coupling matrix, $K$, is positive definite.   This restricted the parameters $p$ and $q$ to satisfy $1+q/p>0$ and $pq>0$, which stringently limits the applicability of the results in \cite{medina}.  In this present work, we study a matrix $K$ of the form given by (\ref{kmat}) satisfying $\det{(K)}=4q/p<0$ and so the system becomes indefinite.  Since the value of $p$ is positive in the derivation of the system \cite{ichinose1997topological,medina}, an indefinite matrix, $K$, will impose the condition that $q<0$.  We then see that the parameter term $1+q/p$ may no longer be positive.  That is, if $q<-p$, then $1+q/p<0$.  One important consequence of this work is that, when combined with \cite{medina}, the full spectrum of $K$ will be covered for doubly-periodic domains.  

When provided with such an indefinite coupling matrix, the difficulty in studying such a system arises after establishing a variational principle.  Normally, one would hope to find a minimizer to the functional, or even a saddle point.  However, such a functional will not have a minimizer and establishing a coercive lower bound becomes impossible.  We are able to modify the approach of solving a constrained minimization problem and then establish partial coercivity.  That is, we select an appropriate admissible class based upon the equation that contributes to the indefinite nature of the functional so that the remaining part of the functional will be coercive.  We are able to use this method of partial coercivity to establish existence of solutions to a regularized system on bounded domains.  Then, we are able to take the limit of the domain to the full-plane.  Prior to doing so, we consider the single equation on a bounded domain and then take that limit.  These solutions will be of finite energy. That is, they satisfy the topological boundary condition  
\begin{equation}
\nonumber u(x)\rightarrow-\ln(2)\qquad v(x)\rightarrow-\ln(2)\qquad\text{ as }\abs{x}\rightarrow\infty.
\end{equation}

Since $K$ is symmetric, in the rest of this paper we will write $k_{11}=k_{22}$ and $k_{12}=k_{21}$.  We will also follow the convention in \cite{medina} by writing $\det{(K)}=\abs{K}$ and by denoting the size of a domain, $\Omega$, by $\abs{\Omega}$.  In the work that follows, it will be necessary for us to separate the possibilities for and indefinite matrix $K$ into the two cases:
\begin{enumerate}
\item If $-4\leq\abs{K}<0$, then $1<k_{12}\leq2$ and $0\leq k_{11}<1.$
\item If $\abs{K}<-4$, then  $k_{12}>2$ and $ k_{11}<0$.
\end{enumerate} 

In the case when $\abs{K}<-4$, we will establish existence of finite energy solutions of (\ref{maineq}) over bounded domains and doubly period domains.  When $-4\leq \abs{K}<0$, we are able to establish the existence of finite energy solutions over the full-plane.  Not only will we establish existence of solutions, but we will also provide an asymptotic analysis of solutions to (\ref{maineq}) for this range of $-4\leq \abs{K}<0$.  


The main results of this paper are presented below.
\begin{thm}\label{asym}
Let $K$ be such that $-4\leq\abs{K}<0$ and let the pair $(\tilde{u},\tilde{v})=(u+\ln2,v+\ln2)\in L^1(\mathbb{R}^2)\times L^1(\mathbb{R}^2)$ be a solution to 
\begin{equation}\label{maineq2}
\begin{cases}
\Delta u=4k_{11}e^u+4k_{12}e^v-4+4\pi\sum\limits_{j=1}^{N_1}\delta_{p_j}(x)\qquad\text{in }\mathbb{R}^2\\
\Delta v=4k_{12}e^u+4k_{11}e^v-4+4\pi\sum\limits_{j=1}^{N_2}\delta_{q_j}(x)\qquad\text{in }\mathbb{R}^2.
\end{cases}
\end{equation}
Then we have the following estimates,
\begin{equation}
\abs{u(x)+\ln2}+\abs{v(x)+\ln2}\leq C\frac{e^{-2\abs{x}}}{\abs{x}^{\frac{1}{2}}}\qquad\text{and}\qquad
\abs{\nabla u(x)}+\abs{\nabla v(x)}\leq C\frac{e^{-2\abs{x}}}{\abs{x}^{\frac{1}{2}}}.
\end{equation}
\end{thm}

\begin{thm}\label{necessary}
The system of equations given by (\ref{maineq}) has a solution over a doubly-periodic domain, $\Omega$, if and only if
\begin{equation}
\abs{\Omega}>\frac{\pi}{2}\left(\abs{\frac{p}{q}}\abs{N_-}+N_+\right),
\end{equation}
where $N_-=N_1-N_2$ and $N_+=N_1+N_2.$
\end{thm}

The structure of the paper is as follows.  In section 2, we present the BPS and vortex equations developed in \cite{medina,ichinose1997topological}.  We also present a discussion on the importance and impact of the parameter, $q$, on the system through the lens of the ``pseudospin'' operator \cite{ichinose1997topological,ichi1, ghosh1998bimerons}.   We will see that this parameter has several effects on the behavior of the system.

In section 3, we will regularize the system (\ref{maineq}) and consider a version of the system with homogenous boundary terms via a shift.  We then establish solutions over a  bounded domain, denoted by $\Omega$, via a variational method. Due to the indefinite nature of the problem, we will use partial coerciveness by ``freezing'' one of the equations in (\ref{maineq}).  We will then take the limit of a single equation version of (\ref{maineq}) to extend it from a bounded domain, $\Omega$, to the full-plane, $\mathbb{R}^2$.  We also take the large domain limit of the entire system from a bounded domain to the full-plane and provide an asymptotic analysis of solutions.

In section 4, we establish the existence of solutions over a doubly-periodic domain, more specifically a torus.  We accomplish this via a minimization of an indefinite functional which satisfies partial coercivity.  We also provide necessary and sufficient conditions on the size of the torus relative to the vortex numbers $N_1$ and $N_2$ for solutions to exist.

\section{BPS and Vortex Equations}\label{BPS}
Here, we provide a brief overview of the system of equations governing the FQHE.   The Lagrangian describing the FQHE in two-dimensional electron systems was well discussed in \cite{ichinose1997topological,medina} and is given by
\begin{equation}
\nonumber\mathcal{L}=\mathcal{L}_{\phi}+\mathcal{L}_{CS},
\end{equation}
which is a sum of the matter term,
\begin{equation}
\nonumber\mathcal{L}_{\phi}=i\bar{\psi}_{\uparrow}(\partial_0-i a_0^+-ia_0^-)\psi_{\uparrow}+i\bar{\psi}_{\downarrow}(\partial_0-ia_0^++ia_0^-)\psi_{\downarrow}
-\frac{1}{2M}\sum_{\sigma=\uparrow,\downarrow}\abs{D_j^{\sigma}\psi_{\sigma}}^2-V(\psi_{\uparrow},\psi_{\downarrow}),
\end{equation}
and the Chern-Simons term,
\begin{equation}
\nonumber\mathcal{L}_{CS}=\mathcal{L}_{CS}(a_{\mu}^+)+\mathcal{L}_{CS}(a_{\mu}^-)=-\frac{1}{4}\epsilon_{\mu\nu\lambda}\left(\frac{1}{p}a_{\mu}^+\partial_{\nu}a_{\lambda}^++\frac{1}{q}a_{\mu}^-\partial_{\nu}a_{\lambda}^-\right).
\end{equation}


The bosonized fields are represented in terms of the upper layer, $\psi_{\uparrow}$, and the lower layer, $\psi_{\downarrow}$.  The mass of the electrons is given by $M$, while the parameters $p$ and $q$ are real numbers.  The scalar potential fields are represented by $a_{\mu}^+$ and $a_{\mu}^-$.  Lastly, the potential between the bosonized electrons is given by,
\begin{equation}\label{potential}
V(\psi_{\uparrow},\psi_{\downarrow})=\frac{p}{M}(\bar{\psi}_{\uparrow}\psi_{\uparrow}+\bar{\psi}_{\downarrow}\psi_{\downarrow})^2
+\frac{q}{M}(\bar{\psi}_{\uparrow}\psi_{\uparrow}-\bar{\psi}_{\downarrow}\psi_{\downarrow})^2+W(\psi_{\uparrow},\psi_{\downarrow}),
\end{equation}
where $W$ represents the long range inter-layer and intra-layer Coulomb repulsions.  In the model discussed by \cite{medina,ichinose1997topological}, this force is assumed to be negligible and will continue to be assumed as such. 

The first integral of the system of BPS type was obtained by Medina in \cite{medina} and it is given by
\begin{align}\label{bps}
(D_1^{\uparrow}&-iD_2^{\uparrow})\psi_{\uparrow}=0\\\label{d1}
(D_1^{\downarrow}&-iD_2^{\downarrow})\psi_{\downarrow}=0\\\label{d2}
B_{12}&=2(p+q)\abs{\psi_{\uparrow}}^2+2(p-q)\abs{\psi_{\downarrow}}^2-eB\\
\tilde{B}_{12}&=2(p-q)\abs{\psi_{\uparrow}}^2+2(p+q)\abs{\psi_{\downarrow}}^2-eB\\
b_0&=\frac{1}{M}(p+q)\abs{\psi_{\uparrow}}^2+\frac{1}{M}(p-q)\abs{\psi_{\downarrow}}^2+\frac{eB}{M}\\\label{d6}
\tilde{b}_0&=\frac{1}{M}(p-q)\abs{\psi_{\uparrow}}^2+\frac{1}{M}(p+q)\abs{\psi_{\downarrow}}^2+\frac{eB}{M},
\end{align}
where equations (\ref{bps}) and (\ref{d1}) are the self dual equations being solved.  There are many solutions to these self dual equations, but this degeneracy is removed by the Chern-Simons constraints \cite{ichinose1997topological}.

 The ground state configurations for the fractional quantum Hall effect are represented in terms of the average electron density, $\bar{\rho}$, $\psi_{\uparrow,0}=\psi_{\downarrow,0}=\sqrt{{\bar{\rho}}/{2}}$.  This yields the following integral representing the total energy of the system \cite{medina},
\begin{equation}\label{energy}
E=\frac{1}{2M}\int \sum \abs{(D_1^{\sigma}-iD_2^{\sigma})\psi_{\sigma}}^2+eB(\abs{\psi_{\uparrow}}^2+\abs{\psi_{\downarrow}}^2)-eB(\abs{\psi_{\uparrow,0}}^2+\abs{\psi_{\downarrow,0}}^2)dx\geq\frac{eB}{4Mp}\Phi_{CS}.
\end{equation}

where the total Chern-Simons flux, $\Phi_{CS}$ is given by
\begin{equation}
\Phi_{CS}=\frac{1}{2}\int(B_{12}+\tilde{B}_{12})dx=\int2p\left(\abs{\psi_{\uparrow}}^2+\abs{\psi_{\downarrow}}^2\right)-eBdx
\end{equation}

Furthermore, in \cite{medina}, solutions to (\ref{bps}) and (\ref{d1}) were considered over the full plane under the condition that they were of finite energy.  As a result, the following boundary conditions were imposed on the complex fields, $\psi_{\uparrow}$ and $\psi_{\downarrow}$,
\begin{equation}\label{gstate}
\abs{\psi_{\uparrow}}^2\rightarrow\abs{\psi_{\uparrow,0}}^2=\frac{\bar{\rho}}{2}\qquad\text{ and }\qquad\abs{\psi_{\downarrow}}^2\rightarrow\abs{\psi_{\downarrow,0}}^2=\frac{\bar{\rho}}{2}\qquad\text{ as }\abs{x}\rightarrow\infty.
\end{equation}

An important observation to make here is that the following terms are affected by $q$: $V,B_{12},\tilde{B}_{12},b_0,\tilde{b}_0$.  Remarkably, the Chern-Simons flux is independent of $q$.  That is, for a given value of $p$, any pair of $p$ and $q$ have the same total Chern-Simons flux.  Moreover, the filling factor, $\nu$, is also independent of $q$ and thus a single filling factor gives rise to a fixed value of $p$, but a wide range of possibilities for $q$.  Since all possible finite energy solutions with $q>0$ was covered in \cite{medina}, we focus solely on $q<0$.  We see that the symmetry of the system implies the approach for existence should be similar to that done in \cite{medina} for $q>0$.  As a matter of fact, we observe that if one were to negate $q$ in (\ref{bps})-(\ref{d6}), the only differences would be that $B_{12}$ and $\tilde{B}_{12}$ switch as well as $b_0$ and $\tilde{b}_0$.  While this may seem like a minor detail, we will observe that the approach taken in \cite{medina} will not suffice.  When establishing a variational principle for $q<0$, the functional becomes indefinite and thus, establishing existence of solutions to (\ref{maineq}) will require a much different approach.

It is also important to note that in \cite{medina}, the following quantization of the magnetic flux integrals was established and will hold regardless of the sign of $q$.  These integrals are
\begin{equation}
\int B_{12}dx=-2\pi p N_1\qquad\text{and}\qquad\int\tilde{B}_{12}dx=-2\pi p N_2
\end{equation}
with the integration begin evaluated over either the doubly-periodic domain or the full-plane.  With this quantization, we obtain
\begin{equation}\label{bsum}
\int \left(B_{12}+\tilde{B}_{12}\right)dx=-2\pi p N_+\qquad\text{and}\qquad \int \left(B_{12}-\tilde{B}_{12}\right)dx=-2\pi p N_-,
\end{equation}
where $N_+=N_1+N_2$ and $N_-=N_1-N_2$.  Using the notation $\Psi=(\psi_{\uparrow},\psi_{\downarrow})^t$, we obtain from (\ref{bsum}) 
\begin{equation}\label{quantq}
\int\left(\bar{\Psi}\Psi-\frac{eB}{2p}\right)dx=-\frac{\pi}{2}N_+\qquad\text{and}\qquad\int\bar{\Psi}\sigma_3\Psi dx=-\frac{\pi p}{2q}N_-.
\end{equation}

In this system, the magnetic field was strong enough to polarize the electrons \cite{medina,ichinose1997topological}.  As a result, a degree of freedom, which was the quantization of spin, is lost.  However,  we regain that degree of freedom and treat the double layer problem as a mono layer problem with a quantization of spin that is dubbed the ``pseudospin'' \cite{ichinose1997topological,ichi1,ghosh1998bimerons}.

Upon further investigation, for the ground state configuration $\Psi_0=\left(\psi_{\uparrow,0},\psi_{\downarrow,0}\right)^t$ given by (\ref{gstate}), we have
\begin{equation}\label{sig3}
\bar{\Psi}_0\sigma_3\Psi_0=0.
\end{equation}
In (\ref{sig3}), $\sigma_3$ is the $z-$component of the pseudospin which represents the difference between the upper and lower layer densities \cite{ghosh1998bimerons} and is given by

\begin{equation*}
\sigma_3=\begin{pmatrix} 
1 & 0 \\
0 & -1 
\end{pmatrix}.
\end{equation*}
As a result, we see that the electric charge of the lower layer, given by $\tilde{Q}$ in \cite{ichinose1997topological,ichi1} satisfies
\begin{equation}
\tilde{Q}=\int\left(\bar{\Psi}\sigma_3\Psi-\bar{\Psi}_0\sigma_3\Psi_0 \right)dx=-\frac{\pi p}{2q}N_-.
\end{equation}
We see that in the indefinite setting, that is when $q<0$, there are three possibilities.  First, we note that if $N_1=N_2$, that is each layer contains the same amount of electrons, then the quantization of the charge in the lower layer, $\tilde{Q}$, satisfies $\tilde{Q}=0$.  Furthermore, when $\tilde{Q}=0$, the Chern-Simons flux, $\tilde{\Phi}_{CS} =2q\tilde{Q}=0$.  
Now we discuss when the number of elections differs in the upper and lower layer.  If $N_1>N_2$ then $N_->0$ and $\tilde{Q}>0$ while $\tilde{\Phi}_{CS}<0$.   If $N_1<N_2$, then $N_-<0$ and the charge satisfies $\tilde{Q}<0$ while the Chern-Simons flux satisfies $\tilde{\Phi}_{CS}>0.$

\section{Topological solutions over the full-plane}
In this section, we establish existence of solutions to (\ref{maineq}) when $\abs{K}<0$ over a bounded domain.  The first step towards establishing existence, is to regularize the system given by (\ref{maineq}) and establish a variational principle as in \cite{lin2007system}.  We will show that the solutions of this regularized system will be critical points of an indefinite action functional and ultimately solutions to (\ref{maineq}) over $\Omega$, which is connected, bounded,  and smooth.  Moreover, $\Omega$ contains the point vortices $p_j$ for $j=1,\ldots, N_1$ and $q_j$ for $j=1,\ldots, N_2$.  

Before proceeding, we recall that the boundary requirement for topological solutions is that $u,v\rightarrow-\ln2$ \cite{ichinose1997topological,medina}.  We consider the substitutions
\begin{equation*}
u_1=u+\ln2\qquad\text{and}\qquad v_1=v+\ln2,
\end{equation*}
to obtain
\begin{equation}\label{s1}
\begin{cases}
\Delta u_1=2k_{11}e^{u_1}+2k_{12}e^{v_1}-4+4\pi\displaystyle\sum\limits_{j=1}^{N_1}\delta_{p_{j}}&\text{in }\Omega\\
\Delta v_1=2k_{12}e^{u_1}+2k_{11}e^{v_1}-4+4\pi\displaystyle\sum\limits_{j=1}^{N_2}\delta_{q_{j}}&\text{in }\Omega.
\end{cases}
\end{equation}
To regularize our system, we use the fact that as $\varepsilon\rightarrow 0$, 
\begin{equation}\label{reg1}
\frac{4\varepsilon}{\left(\varepsilon+\abs{x-p_j}^2\right)^2}\xrightharpoonup{*}4\pi\delta p_j\qquad\text{ and }\qquad \frac{4\varepsilon}{\left(\varepsilon+\abs{x-q_j}^2\right)^2}\xrightharpoonup{*}4\pi\delta q_j.
\end{equation}
The convergence above is in sense of distibution.
In view of (\ref{reg1}), we consider the regularized form of (\ref{s1}),
\begin{equation}\label{s2}
\begin{cases}
\Delta u_1=2k_{11}e^{u_1}+2k_{12}e^{v_1}-4+\sum\limits_{j=1}^{N_1}\frac{4\varepsilon}{\left(\varepsilon+\abs{x-p_j}^2\right)^2}&\text{in }\Omega\\
\Delta v_1=2k_{12}e^{u_1}+2k_{11}e^{v_1}-4+\sum\limits_{j=1}^{N_2}\frac{4\varepsilon}{\left(\varepsilon+\abs{x-q_j}^2\right)^2}&\text{in }\Omega,
\end{cases}
\end{equation}
and introduce the subtractive background functions,
\begin{equation}\label{backuv}
u_0^{\varepsilon}=\sum\limits_{j=1}^{N_1}\ln\left(\frac{\varepsilon+\abs{x-p_j}^2}{1+\abs{x-p_j}^2} \right)\qquad
v_0^{\varepsilon}=\sum\limits_{j=1}^{N_2}\ln\left(\frac{\varepsilon+\abs{x-q_j}^2}{1+\abs{x-q_j}^2} \right),
\end{equation}
such that for some $h_1, h_2\in W^{1,2}(\Omega)$, independent of $\varepsilon>0$, we have,
\begin{equation}
\Delta u_0^{\varepsilon}=-h_1+\sum\limits_{j=1}^{N_1}\frac{4\varepsilon}{\left(\varepsilon+\abs{x-p_j}^2\right)^2}\qquad\text{and}\qquad
\Delta v_0^{\varepsilon}=-h_2+\sum\limits_{j=1}^{N_2}\frac{4\varepsilon}{\left(\varepsilon+\abs{x-q_j}^2\right)^2}.
\end{equation}
With these background functions constructed, we are able to set
\begin{equation}\label{u2v2}
u_1=u_0^{\varepsilon}+u_2\qquad\text{ and }\qquad v_1=v_0^{\varepsilon}+v_2,
\end{equation}
then rewrite (\ref{s2}) in terms of $u_2$ and $v_2$.  In view of (\ref{u2v2}), the equations in (\ref{s2}) become
\begin{equation}\label{s3}
\begin{cases}
\Delta u_2=2k_{11}e^{u_0^{\varepsilon}+u_2}+2k_{12}e^{v_0^{\varepsilon}+v_2}-4+h_1& \text{in }\Omega\\
\Delta v_2=2k_{12}e^{u_0^{\varepsilon}+u_2}+2k_{11}e^{v_0^{\varepsilon}+v_2}-4+h_2& \text{in }\Omega\\
u_2=v_2=0\qquad\text{on }\partial\Omega
\end{cases}
\end{equation}
We will find it useful to decompose $u_2$ and $v_2$ into a sum.  To this end, we write $u_2$ as a sum of a harmonic function and another function.  We repeat this for  $v_2$ as well.  We let $U_0^{\varepsilon}$ and $V_0^{\varepsilon}$ be harmonic in $\Omega$ and satisfy the following boundary conditions
\begin{equation}
U_0^{\varepsilon}=-u_0^{\varepsilon}\qquad\text{and}\qquad V_0^{\varepsilon}=-v_0^{\varepsilon}\qquad\text{on}\qquad \partial\Omega.  
\end{equation}
Then, for some functions $u_3$ and $v_3$, we write
\begin{equation}\label{2harm}
u_2=U_0^{\varepsilon}+u_3\qquad
\text{and}\qquad v_2=V_0^{\varepsilon}+v_3.
\end{equation}
In view of these transformations, the system in (\ref{s3}) can now be written in terms of $u_3$ and $v_3$ as the following homogenous boundary value problem,
\begin{equation}\label{s41}
\begin{cases}
\Delta u_3=2k_{11}e^{u_0^{\varepsilon}+U_0^{\varepsilon}+u_3}+2k_{12}e^{v_0^{\varepsilon}+V_0^{\varepsilon}+v_3}-4+h_1& \text{in }\Omega\\
\Delta v_3=2k_{12}e^{u_0^{\varepsilon}+U_0^{\varepsilon}+u_3}+2k_{11}e^{v_0^{\varepsilon}+V_0^{\varepsilon}+v_3}-4+h_2& \text{in }\Omega\\
u_3=v_3=0\qquad \text { on }\partial\Omega.
\end{cases}
\end{equation}
In the final substitution to transform this system into one that has a variational principle, we use
\begin{equation*}
f_0^{\varepsilon}=u_0^{\varepsilon}+U_0^{\varepsilon},\quad
g_0^{\varepsilon}=v_0^{\varepsilon}+V_0^{\varepsilon}
\xi=u_3+v_3, \text{ and }
\zeta=u_3-v_3.
\end{equation*}
This substitution yields the following system,
\begin{equation}\label{s4}
\begin{cases}
\Delta\xi=4e^{f_0^{\varepsilon}+\frac{1}{2}(\xi+\zeta)}+4e^{g_0^{\varepsilon}+\frac{1}{2}(\xi-\zeta)}-8+(h_1+h_2)& \text{in }\Omega\\
\Delta\zeta=\abs{K}e^{f_0^{\varepsilon}+\frac{1}{2}(\xi+\zeta)}-\abs{K}e^{g_0^{\varepsilon}+\frac{1}{2}(\xi-\zeta)}+(h_1-h_2)&\text{in }\Omega\\
\xi=\zeta=0\qquad\text{on }\partial\Omega,
\end{cases}
\end{equation} 
where we have used properties of the coupling matrix to obtain $k_{11}+k_{12}=2$ and $k_{11}-k_{12}={\abs{K}}/{2}.$  We then multiply the first equation in (\ref{s4}) by ${\abs{K}}/{4}$ to obtain,
\begin{equation}\label{s5}
\begin{cases}
\frac{\abs{K}}{4}\Delta\xi=\abs{K}e^{f_0^{\varepsilon}+\frac{1}{2}(\xi+\zeta)}+\abs{K}e^{g_0^{\varepsilon}+\frac{1}{2}(\xi-\zeta)}-2\abs{K}+\frac{\abs{K}}{4}(h_1+h_2)\quad\text{in }\Omega\\
\Delta\zeta=\abs{K}e^{f_0^{\varepsilon}+\frac{1}{2}(\xi+\zeta)}-\abs{K}e^{g_0^{\varepsilon}+\frac{1}{2}(\xi-\zeta)}+(h_1-h_2)\qquad\text{in }\Omega\\
\xi=\zeta=0\qquad\text{on }\partial\Omega
\end{cases}
\end{equation} 
which allows us to see that these are the Euler-Lagrange equations of the indefinite functional,
\begin{multline}\label{f5}
I(\xi,\zeta)=\int\displaylimits_{\Omega}\bigg[\frac{\abs{K}}{8}\abs{\nabla\xi}^2+\frac{1}{2}\abs{\nabla\zeta}^2+2\abs{K}e^{f_0^{\varepsilon}+\frac{1}{2}(\xi+\zeta)}+2\abs{K}e^{g_0^{\varepsilon}+\frac{1}{2}(\xi-\zeta)}\\-2\abs{K}\xi+\frac{\abs{K}}{4}(h_1+h_2)\xi+(h_1-h_2)\zeta\bigg] dx.
\end{multline}

In this scenario, a direct minimization approach will not work because we are unable to establish coercivity.  That is, the negative coefficient of $\abs{\nabla\xi}$ creates an issue.  Also, we can not establish a compactness condition that would allow for saddle point solutions via a mountain pass theorem.  Therefore, we take an approach similar to \cite{benci1979critical,lin2007system,yang2013solitons} and separate the functional (\ref{f5}) into two parts, the negative and positive and then ``freeze'' the negative part by introducing a constraint which is the weak formulation of the first equation in (\ref{s4}).  With this in mind, we will consider the following minimization problem,
\begin{equation}\label{indefmin}
\min\left\{I(\xi,\zeta)|(\xi,\zeta)\in\mathcal{C}\right\}
\end{equation}
where the admissible class, $\mathcal{C}$, is defined by
\begin{equation}\label{admis}
\mathcal{C}=\left\{(\xi,\zeta)|\xi,\zeta\in W^{1,2}_0(\Omega)\text{ and }\xi,\zeta\text{ satisfy (i)}\right\}.
\end{equation}
\begin{align*}
\text{(i)}\quad&\int\displaylimits_{\Omega}\bigg[\nabla\xi\cdot\nabla w+4we^{f_0^{\varepsilon}+
\frac{1}{2}(\xi+\zeta)}+4we^{g_0^{\varepsilon}+\frac{1}{2}(\xi-\zeta)}-8w+(h_1+h_2)w\bigg]dx=0,\qquad\forall w\in W^{1,2}_0(\Omega).
\end{align*}
We remark here that definition (i) is nothing other than the weak formulation of the first equation in (\ref{s4}) and now state the main theorem of this section.  

\begin{thm}\label{tm21}
There exists a pair $(\xi,\zeta)\in\mathcal{C}$ that is a solution to the minimization problem (\ref{indefmin}).
\end{thm}
Prior to proving Theorem \ref{tm21}, we establish some necessary lemmas.

\begin{lem}
Definition (i) is well posed.  In other words, for any $\zeta\in W^{1,2}_0(\Omega)$ there exists a unique $\xi\in W^{1,2}_0(\Omega)$ satisfying (i).  Furthermore, $\xi$ is the global minimizer of the functional,
\begin{multline}\label{jfun}
J_{\zeta}(\xi)=\int\displaylimits_{\Omega}\bigg[-\frac{\abs{K}}{8}\abs{\nabla\xi}^2-2\abs{K}e^{f_0^{\varepsilon}+
\frac{1}{2}(\xi+\zeta)}-2\abs{K}e^{g_0^{\varepsilon}+\frac{1}{2}(\xi-\zeta)}+2\abs{K}\xi-\frac{\abs{K}}{4}(h_1+h_2)\xi \bigg]dx,
\end{multline}
in $W^{1,2}_0(\Omega)$.
\end{lem}
\begin{proof}
We prove this lemma in three steps.  The first is to show that the functional, $J_{\zeta}(\cdot)$ is weakly lower semicontinuous.  Then, we must show that $J_{\zeta}(\cdot)$ is coercive.  Finally, we must determine uniqueness.

First, we note that the functional (\ref{jfun}) can be rewritten using the $L^2$ norm.  This can be seen below
\begin{multline}\label{jfun2}
J_{\zeta}(\xi)=-\frac{\abs{K}}{8}\|\nabla\xi\|_{L^2(\Omega)}^2-\int\displaylimits_{\Omega}\bigg[2\abs{K}e^{f_0^{\varepsilon}+
\frac{1}{2}(\xi+\zeta)}-2\abs{K}e^{g_0^{\varepsilon}+\frac{1}{2}(\xi-\zeta)}+2\abs{K}\xi-\frac{\abs{K}}{4}(h_1+h_2)\xi \bigg]dx
\end{multline}
We then recall that the $L^p$ norm is weakly lower semicontinuous.  Therefore, in order to establish that the functional $J_{\zeta}(\cdot)$ is weakly lower semicontinuous, we only need to worry about the last four terms in (\ref{jfun2}).  We will first establish continuity of the last term.  To this end, we make the observation that if $\xi_n\rightarrow\xi$ in $L^2(\Omega)$, then by way of H\"older's inequality  we obtain,
\begin{align}\label{contin1}
\begin{split}
\abs{\int\displaylimits_{\Omega}\bigg[(h_1+h_2)\xi-(h_1+h_2)\xi_n\bigg]dx}&\leq\|h_1+h_2\|_{L^2(\Omega)}\|\xi-\xi_n\|_{L^2(\Omega)}.
\end{split}
\end{align}

Since $h_1,h_2$ are in $W^{1,2}(\Omega)$, we see that they are bounded in $L^2(\Omega)$.  We also have $\xi_n\rightarrow\xi$ in $L^2(\Omega)$ by our assumption.  Therefore, we obtain that the right hand side of (\ref{contin1}) vanishes as $n\rightarrow\infty$ and consequently, the last term in (\ref{jfun2}) is continuous.  By a similar argument, we obtain continuity for $\int\displaylimits2\abs{K}\xi dx$.  It is obvious that the remaining terms of $J_{\zeta}(\xi)$ are continuous and thus $J_{\zeta}$ is weakly lower semicontinuous.

We now look to establish coercivity of the functional $J_{\zeta}(\xi)$.   Here, we recall that $-2\abs{K}>0$ and therefore, the integrals involving the exponential terms are positive.  This means we have
\begin{equation}\label{jbound1}
J_{\zeta}(\xi)\geq-\frac{\abs{K}}{8}\|\nabla\xi\|_{L^2(\Omega)}^2+2\abs{K}\int\displaylimits_{\Omega}\xi dx-\frac{\abs{K}}{4}\int\displaylimits_{\Omega}(h_1+h_2)dx,
\end{equation}
and now recall the Poincar\'e inequality on $W^{1,2}_0(\Omega)$,
\begin{equation}\label{poin11}
\|\xi\|_{L^2(\Omega)}\leq C\|\nabla\xi\|_{L^2(\Omega)}.
\end{equation}
We use Cauchy's inequality and (\ref{poin11}) to obtain the following upper bound for the last term in (\ref{jbound1}),
\begin{align}\label{bound11}
\abs{\int\displaylimits_{\Omega}(h_1+h_2)\xi dx}&
=\frac{1}{4\varepsilon}\|h_1+h_2\|_{L^2(\Omega)}^2+\varepsilon\|\xi\|_{L^2(\Omega)}^2\leq C_1(\varepsilon,h_1,h_2)+C\varepsilon\|\nabla\xi\|_{L^2(\Omega)}^2.
\end{align}
Similarly, we obtain
\begin{equation}\label{bound11a}
\abs{\int\displaylimits_{\Omega}\xi dx}\leq C_2(\varepsilon,\abs{\Omega})+C\varepsilon\|\nabla\xi\|_{L^2(\Omega)}^2.
\end{equation}
The upper bounds obtained in (\ref{bound11}) and (\ref{bound11a}) allow us to obtain the following lower bound for $J_{\zeta}(\xi),$
\begin{equation}
J_{\zeta}(\xi)\geq\left(-\frac{\abs{K}}{8}+\frac{9\abs{K}}{4}\varepsilon C\right)\|\nabla\xi\|_{L^2(\Omega)}^2-C_3(\varepsilon,h_1,h_2,\abs{\Omega}).
\end{equation}
We recall that $-\abs{K}>0$ and so we can choose $\varepsilon>0$ small enough so that
\begin{equation}
-\frac{\abs{K}}{8}+\frac{9\abs{K}}{4}\varepsilon C>0,
\end{equation}
we see that $J_{\zeta}(\xi)$ is coercive and bounded below.  This, together with the fact that $J_{\zeta}(\cdot)$ is weakly lower semicontinuous, implies that the minimizer exists.  The convexity of $J_{\zeta}(\xi)$ in terms of $\xi$ implies that that the minimizer is unique.
\end{proof}

For the next lemma, it is useful to write the functional in (\ref{f5}) as
\begin{equation}\label{ij}
I(\xi,\zeta)=\frac{1}{2}\|\nabla\zeta\|_{L^2(\Omega)}^2+\int\displaylimits_{\Omega}(h_1-h_2)\zeta dx-J_{\zeta}(\xi),
\end{equation}
so that we may establish partial coerciveness below.
\begin{lem}
The functional, $I(\xi,\zeta)$ given by (\ref{f5}) and ultimately (\ref{ij}) is partially coercive.
\end{lem}
\begin{proof}
First, we note that since $J_{\zeta}(\xi)$ is minimized (uniquely) by $\xi$, then by the maximal growth function for $W_0^{1,2}(\Omega)$ \cite{trudinger1967imbeddings}, we have for any $\zeta\in W_0^{1,2}(\Omega)$
\begin{equation}\label{jmin}
J_{\zeta}(\xi)\leq J_{\zeta}(0)\nonumber=-2\abs{K}\int\displaylimits_{\Omega}\bigg[e^{f_0^{\varepsilon}+\frac{1}{2}\zeta}+e^{g_0^{\varepsilon}-\frac{1}{2}\zeta}\bigg]dx
\leq C\int\displaylimits_{\Omega}\bigg[e^{\frac{1}{2}\zeta}+e^{-\frac{1}{2}\zeta}\bigg]dx
\leq C\int\displaylimits_{\Omega}\cosh\left(\frac{\zeta}{2}\right)dx
\leq C'.
\end{equation}

In view of (\ref{bound11}) and (\ref{jmin}),  we are able to obtain the following lower bound for $I(\xi,\zeta)$,
\begin{equation}\label{imin1}
I(\xi,\zeta)\geq\left(\frac{1}{2}-C\varepsilon\right)\|\nabla\zeta\|_{L^2(\Omega)}^2-C''
\end{equation}
where $C, C''>0$.  We again choose $\varepsilon$ small enough so that the coefficient of $\|\nabla\zeta\|_{L^2(\Omega)}^2$ is positive and so we may now conclude that the functional $I(\xi,\zeta)$ is both partially coercive and bounded below.  Meaning, since we are able to ``freeze'' the negative part of the functional, we are then able to establish coercivity for $I(\xi,\zeta)$.
\end{proof}

Now that partial coercivity of the functional $I(\xi,\zeta)$ has been established, we establish convergence of a minimizing sequence of (\ref{indefmin}).
\begin{lem}\label{weakmin}
A minimizing sequence $\left\{(\xi_n,\zeta_n)\right\}_{n=1}^{\infty}$ of (\ref{indefmin}) satisfies $\xi_n\rightharpoonup\xi$ weakly in $W^{1,2}_0(\Omega)$ and $\zeta_n\rightharpoonup\zeta$ weakly in $W^{1,2}_0(\Omega)$ as $n\rightarrow\infty$.
\end{lem}

\begin{proof}
We let $\left\{(\xi_n,\zeta_n)\right\}_{n=1}^{\infty}$ be any minimizing sequence of (\ref{indefmin}) such that,
\begin{equation*}
I(\xi_1,\zeta_1)\geq I(\xi_2,\zeta_2)\geq\cdots\geq I(\xi_n,\zeta_n)\geq\cdots.
\end{equation*}
We then define
\begin{equation}
\eta_0:=\inf\left\{I(\xi,\zeta)\mid (\xi,\zeta)\in\mathcal{C}\right\}=\lim\limits_{n\rightarrow\infty}I(\xi_n,\zeta_n).
\end{equation}
It is easy to see that by (\ref{imin1}), the sequence $\left\{\zeta_n\right\}_{n=1}^{\infty}$ satisfies,
\begin{equation}\label{bound3}
M_0\geq I(\xi_n,\zeta_n)\geq \left(\frac{1}{2}-\varepsilon C\right)\|\nabla\zeta_n\|_{L^2(\Omega)}^2-C'',
\end{equation}
where the upper bound, $M_0$, comes from the fact that a minimizing sequence is bounded.  By (\ref{bound3}), we see that $\nabla\zeta_n$ is bounded in $L^2(\Omega)$.  Furthermore, the Poincar\'e inequality on $W^{1,2}_0(\Omega)$ gives
\begin{equation}\label{poin1}
\|\zeta_n\|_{L^2(\Omega)}^2\leq C\|\nabla\zeta_n\|_{L^2(\Omega)}^2.
\end{equation}
When we insert (\ref{poin1}) into (\ref{imin1}), we establish that $\left\{\zeta_n\right\}_{n=1}^{\infty}$ is bounded in $L^2(\Omega)$.  Notice that we get
\begin{equation}\label{in46}
M_0\geq C_3\|\zeta_n\|_{L^2(\Omega)}^2-C_2
\end{equation}
for some $C_3>0$.  Inequality (\ref{in46}) implies that $\left\{\zeta_n\right\}_{n=1}^{\infty}$ is bounded in $L^2(\Omega)$, and ultimately bounded in  $W^{1,2}_0(\Omega)$. 

 In a similar fashion, we use (\ref{jmin}) to obtain
\begin{align}\label{in47}
M_0'\geq C_4\|\nabla\xi_n\|_{L^2(\Omega)}^2-C_1\geq C_5\|\xi_n\|_{L^2(\Omega)}^2-C_1
\end{align}
where $C_4,C_5>0$.  Inequality (\ref{in47}) establishes the boundedness of $\left\{\xi_n\right\}_{n=1}^{\infty}$ in $W^{1,2}_0(\Omega)$.  We then conclude that (passing to a subsequence if necessary) $\xi_n\rightharpoonup \xi$ weakly in $W^{1,2}_0(\Omega)$.

\end{proof}
We are now ready show that the pair of functions $(\xi,\zeta)$ found in Lemma \ref{weakmin} belongs to the admissible class given by (\ref{admis}).
\begin{lem}
The functions, $\xi$ and $\zeta$ defined in Lemma \ref{weakmin} satisfy $\left(\xi,\zeta\right)\in\mathcal{C}$.  
\end{lem}
\begin{proof}
Since $\xi_n,\zeta_n\in\mathcal{C}$, we know that these satisfy condition (i) in (\ref{admis}). We can write this condition as the following
\begin{multline}\label{18}
\int\displaylimits_{\Omega}\bigg[\nabla\xi_n\cdot\nabla w+4we^{f_0^{\varepsilon}+
\frac{1}{2}(\xi_n+\zeta_n)}+4we^{g_0^{\varepsilon}+\frac{1}{2}(\xi_n-\zeta_n)}-8w+(h_1-h_2)w\bigg]dx=0,
\quad\forall w\in W^{1,2}_0(\Omega).
\end{multline}
Now, we recall that $\xi_n\rightharpoonup\xi$ and $\zeta_n\rightharpoonup\zeta$ weakly in $W^{1,2}_0(\Omega)$.  From this weak convergence, we know that we can find subsequences, denoted by $\left\{\xi_{n_j}\right\}_{j=1}^{\infty}$ and $\left\{\zeta_{n_j}\right\}_{j=1}^{\infty}$, such that they converge strongly to $\xi$ and $\zeta$ in $L^2(\Omega)$.  Without loss of generality (passing to a subsequence if necessary), we will assume that $\xi_n\rightarrow\xi$, and $\zeta_n\rightarrow\zeta$ strongly in $L^2(\Omega)$ as $n\rightarrow\infty$. 

We recall the Trudinger-Moser inequality of the form,
\begin{equation}\label{doubleTM}
\int\displaylimits_{\Omega}e^fdx\leq C_1'e^{C_2'\|f\|_{W^{1,2}(\Omega)}^2},\qquad f\in W^{1,2}(\Omega).
\end{equation}
We will use the inequality given by (\ref{doubleTM}) to show that, as $n\rightarrow\infty$, we obtain the convergence of,
\begin{equation}\label{e1}
\int\displaylimits_{\Omega}e^{f_0^{\varepsilon}+\frac{1}{2}(\xi_n+\zeta_n)}dx\rightarrow\int\displaylimits_{\Omega}e^{f_0^{\varepsilon}+\frac{1}{2}(\xi+\zeta)}dx,
\end{equation}
and
\begin{equation}\label{e2}
\int\displaylimits_{\Omega}e^{g_0^{\varepsilon}+\frac{1}{2}(\xi_n-\zeta_n)}dx\rightarrow\int\displaylimits_{\Omega}e^{g_0^{\varepsilon}+\frac{1}{2}(\xi-\zeta)}dx.
\end{equation}
In order to establish both (\ref{e1}) and (\ref{e2}), we will also use the mean value theorem and the general H\"older inequality.  With these tools, we are able to obtain the following estimates,
\begin{align}
&\nonumber\abs{\int\displaylimits_{\Omega}e^{f_0^{\varepsilon}+\frac{1}{2}(\xi_n+\zeta_n)}dx-\int\displaylimits_{\Omega}e^{f_0^{\varepsilon}+\frac{1}{2}(\xi+\zeta)}dx}\leq C\int\displaylimits_{\Omega}e^{\frac{\abs{\xi_n}+\abs{\xi}}{2}}e^{\frac{\abs{\zeta_n}+\abs{\zeta}}{2}}\abs{\zeta_n-\zeta+\xi_n-\xi}dx\\
\nonumber&\leq C\left(\int\displaylimits_{\Omega}  e^{4\abs{\xi_n}} dx\right)^{\frac{1}{8}}\left(\int\displaylimits_{\Omega}e^{4\abs{\xi}}   dx\right)^{\frac{1}{8}}\left(\int\displaylimits_{\Omega}  e^{4\abs{\zeta_n}} dx\right)^{\frac{1}{8}}\left(\int\displaylimits_{\Omega} e^{4\abs{\zeta}}  dx\right)^{\frac{1}{8}}\|\zeta_n-\zeta+\xi_n-\xi\|_{L^2(\Omega)}\\
&\leq C'\left(\|\zeta_n-\zeta\|_{L^2(\Omega)}+\|\xi_n-\xi\|_{L^2(\Omega)} \right)\rightarrow 0\text{ as }n\rightarrow\infty.
\end{align}

One can see that we have now established the convergence given by (\ref{e1}).  In a similar fashion, we establish (\ref{e2}). The details are left to the reader as they are nearly a mirror image of the verification of (\ref{e1}).  We take $n\rightarrow\infty$ in (\ref{18}) and see that $(\xi,\zeta)\in\mathcal{C}$. 
\end{proof}
\begin{lem}
The convergence of $\xi_n\rightarrow\xi$ is strong convergence in $W^{1,2}_0(\Omega)$.
\end{lem}
\begin{proof}
First, we let $w=\xi_n-\xi$ in (i) and in (\ref{18}).  With this substitution, we end up with both
\begin{multline}\label{z1}
\int\displaylimits_{\Omega}\bigg[\nabla\xi\cdot\nabla  (\xi_n-\xi)+4 (\xi_n-\xi)e^{f_0^{\varepsilon}+
\frac{1}{2}(\xi+\zeta)}+4 (\xi_n-\xi)e^{g_0^{\varepsilon}+\frac{1}{2}(\xi-\zeta)}\\-8 (\xi_n-\xi)+(h_1+h_2) (\xi_n-\xi)\bigg]dx=0
\end{multline}
and
\begin{multline}\label{z2}
\int\displaylimits_{\Omega}\bigg[\nabla\xi_n\cdot\nabla (\xi_n-\xi)+4(\xi_n-\xi)e^{f_0^{\varepsilon}+
\frac{1}{2}(\xi_n+\zeta_n)}+4(\xi_n-\xi)e^{g_0^{\varepsilon}+\frac{1}{2}(\xi_n-\zeta_n)}\\-8(\xi_n-\xi)+(h_1-h_2)(\xi_n-\xi)\bigg]dx=0.
\end{multline}
We now subtract (\ref{z2}) from (\ref{z1}) and take the limit as $n\rightarrow\infty$ to observe that
\begin{multline}
\int\displaylimits_{\Omega}\abs{\nabla\xi_n-\nabla\xi}^2dx=\int\displaylimits_{\Omega}\bigg[4(\xi_n-\xi)e^{f_0^{\varepsilon}}\left(e^{\frac{1}{2}(\xi+\zeta)}-e^{\frac{1}{2}(\xi_n+\zeta_n)}\right)\\+4(\zeta_n-\zeta)e^{g_0^{\varepsilon}}\left(e^{\frac{1}{2}(\xi-\zeta)}-e^{\frac{1}{2}(\xi_n-\zeta_n)} \right)\bigg]dx\rightarrow 0.
\end{multline}
Therefore, we have shown that $\nabla\xi_n\rightarrow\nabla\xi$ in $L^2(\Omega)$.  We may conclude now that $\xi_n\rightarrow\xi$ strongly in $W^{1,2}_0(\Omega)$.
\end{proof}
It is at this point that we are able to provide the proof of Theorem \ref{tm21}. 

\begin{proof}
We recall that $\zeta_n\rightharpoonup\zeta$ weakly in $W^{1,2}_0(\Omega)$ and $\xi_n\rightarrow\xi$ strongly in $W^{1,2}_0(\Omega).$ As a result, we have
\begin{equation}\label{419}
\lim\limits_{n\rightarrow\infty}J_{\zeta_n}(\xi_n)=J_{\zeta}(\xi).
\end{equation}
We use (\ref{ij}), (\ref{419}), and recall that $\displaystyle\eta_0=\lim\limits_{n\rightarrow\infty}I(\xi_n,\zeta_n)$ to obtain the estimate
\begin{align}
\eta_0&=\lim\limits_{n\rightarrow\infty}I(\xi_n,\zeta_n)\geq \frac{1}{2}\|\nabla\zeta\|_{L^2(\Omega)}^2+\int\displaylimits_{\Omega}(h_1-h_2)\zeta dx-J_{\zeta}(\xi)\geq I(\xi,\zeta).
\end{align}
Since we have already established that $(\xi,\zeta)\in\mathcal{C}$, we see that the pair solves the minimization problem given by (\ref{indefmin}).
\end{proof}

Next, we establish that the functions we found are solutions to the system (\ref{s4}).
\begin{lem}\label{chainrule}
The pair $(\xi,\zeta)$ defined in Lemma \ref{weakmin} is a solution to the system given by (\ref{s4}).  
\end{lem}
\begin{proof}
Before anything else, we point out that the weak form of the second equation in (\ref{s4}) is indeed the condition (i) in (\ref{admis}).  Therefore, in the present, we only need to verify the first equation in (\ref{s4}).  To this end, we let $\hat{\zeta}$ be any test function in $W^{1,2}_0(\Omega)$.  We then set $\zeta_{\tau}=\zeta+\tau\hat{\zeta}$.  We have shown that for any $\zeta$ there is a unique minimizer $\xi$ of $J_{\zeta}(\xi)$, we denote the unique minimizer of $J_{\zeta_{\tau}}(\cdot)$ by $\xi_{\tau}$.  We remark here on the dependence of $\xi$ on $\zeta$ and thus $\xi_{\tau}$ on $\zeta_{\tau}$.  That is, we have $\xi_{\tau}=\xi_{\tau}(\zeta+\tau\hat{\zeta})$ and $\xi_{\tau}$ depends smoothly upon $\tau$.  We then set
\begin{equation}\label{xizeta}
\hat{\xi}=\left(\frac{d}{d\tau}\xi_{\tau}\right)\Bigg|_{\tau=0}.
\end{equation}
  Since the pair $(\xi,\zeta)$ minimizes $I(\cdot,\cdot)$, we see that $I(\xi_{\tau},\zeta_{\tau})$ attains its minimum at $\tau=0$.  Therefore, we have
\begin{equation}\label{id}
\frac{d}{d\tau}I(\xi_{\tau},\zeta_{\tau})\Big|_{\tau=0}=0.
\end{equation}
Then, in view of the formulation of $I(\xi,\zeta)$ given by (\ref{f5}), we rewrite (\ref{id}) as
\begin{multline}\label{672}
\int\displaylimits_{\Omega}\bigg[\frac{\abs{K}}{4}\nabla\xi\cdot\nabla\left(\frac{d}{d\tau}\xi_{\tau}\right)\Bigg|_{\tau=0}+\nabla\zeta\cdot\nabla\hat{\zeta}+\abs{K}e^{f_0^{\varepsilon}+\frac{1}{2}(\xi+\zeta)}\hat{\xi}\left(\left(\frac{d}{d\tau}\xi_{\tau}\right)\Bigg|_{\tau=0}+\hat{\zeta}\right)\\+\abs{K}e^{g_0^{\varepsilon}+\frac{1}{2}(\xi-\zeta)}\hat{\xi}\left(\left(\frac{d}{d\tau}\xi_{\tau}\right)\Bigg|_{\tau=0}-\hat{\zeta}\right)-\left(2\abs{K}-\frac{\abs{K}}{4}(h_1+h_2)\right)\left(\frac{d}{d\tau}\xi_{\tau}\right)\Bigg|_{\tau=0}\\+(h_1-h_2)\hat{\zeta} \bigg]dx=0.
\end{multline}
We may use (\ref{xizeta}) to rewrite (\ref{672}) in the more convenient form,

\begin{multline}
\int\displaylimits_{\Omega}\bigg[\nabla\zeta\cdot\nabla\hat{\zeta}+\abs{K}e^{f_0^{\varepsilon}+\frac{1}{2}(\xi+\zeta)}\hat{\zeta}-\abs{K}e^{g_0^{\varepsilon}+\frac{1}{2}(\xi-\zeta)}\hat{\zeta}+(h_1-h_2)\hat{\zeta} \bigg]dx\\=\int\displaylimits_{\Omega}\bigg[-\frac{\abs{K}}{4}\nabla\xi\cdot\nabla\hat{\xi}-\abs{K}e^{f_0^{\varepsilon}+\frac{1}{2}(\xi+\zeta)}\hat{\xi}-\abs{K}e^{g_0^{\varepsilon}+\frac{1}{2}(\xi-\zeta)}\hat{\xi}+2\abs{K}\hat{\xi}-\frac{\abs{K}}{4}(h_1+h_2)\hat{\xi}\bigg]dx
\end{multline}
By the condition (i) in (\ref{admis}), we see that the right hand side is zero.  Therefore, we have obtained
\begin{equation}
\int\displaylimits_{\Omega}\nabla\zeta\cdot\nabla\hat{\zeta}+\abs{K}e^{f_0^{\varepsilon}+\frac{1}{2}(\xi+\zeta)}\hat{\zeta}-\abs{K}e^{g_0^{\varepsilon}+\frac{1}{2}(\xi-\zeta)}\hat{\zeta}+(h_1-h_2)\hat{\zeta} dx=0
\end{equation}
for an arbitrary test function $\hat{\zeta}\in W^{1,2}_0(\Omega)$.  It is clear that this is the weak formulation of the first equation in (\ref{s4}) and we have fully verified the system (\ref{s4}).  
\end{proof}

Now that we have established existence of solutions to the regularized system in (\ref{s3}), we must go back in terms of the original variables.  To do so, we first see that we have found a solution pair to the regularized system (\ref{s2}) which we can denote by $(u_{\varepsilon},v_{\varepsilon})$.  We may now discuss the convergence of such solutions.
\begin{prop}\label{epsilonbound}
When $-4\leq\abs{K}<0$, the solution pair $(u_{\varepsilon},v_{\varepsilon})$ satisfies $u_{\varepsilon}, v_{\varepsilon}<-\ln 2$ in $\Omega$.
\end{prop}
\begin{proof}
Let $(u_{\varepsilon},v_{\varepsilon})$ be a solution to the regularized system (\ref{s2}) where $\varepsilon>0$ is a small positive number. We now write (\ref{s2}) in terms of this pair,
\begin{equation}\label{s2'}
\begin{cases}
\Delta u_{\varepsilon}=4k_{11}e^{u_{\varepsilon}}+4k_{12}e^{v_{\varepsilon}}-4+\sum\limits_{j=1}^{N_1}\frac{4\varepsilon}{\left(\varepsilon+\abs{x-p_j}^2\right)^2}&\text{in }\Omega\\
\Delta v_{\varepsilon}=4k_{12}e^{u_{\varepsilon}}+4k_{11}e^{v_{\varepsilon}}-4+\sum\limits_{j=1}^{N_2}\frac{4\varepsilon}{\left(\varepsilon+\abs{x-q_j}^2\right)^2}&\text{in }\Omega\\\
 u_{\varepsilon}=v_{\varepsilon}=-\ln 2\qquad \text{on }\partial\Omega
\end{cases}
\end{equation}
and observe that we can take $\varepsilon$ small enough so that 
\begin{equation}\label{delta4}
\sum\limits_{j=1}^{N_1}\frac{4\varepsilon}{\left(\varepsilon+\abs{x-p_j}^2\right)^2}>4\qquad\text{and}\qquad\sum\limits_{j=1}^{N_2}\frac{4\varepsilon}{\left(\varepsilon+\abs{x-q_j}^2\right)^2}>4.
\end{equation}
  To see this, observe that when $x=p_j$ or $x=q_j$, the fractions in the summation become ${4}/{\varepsilon}$.  So, as long as $\varepsilon<1$, we are guaranteed to satisfy (\ref{delta4}) for any values of $N_1,N_2\geq 1$. We also see that 
\begin{align}
4k_{11}e^{u_{\varepsilon}}+4k_{12}e^{v_{\varepsilon}}\geq0\qquad\text{and}\qquad
4k_{12}e^{u_{\varepsilon}}+4k_{11}e^{v_{\varepsilon}}\geq0.
\end{align}
This allows us to estimate (\ref{s2'}) as
\begin{equation}\label{s2''}
\begin{cases}
\Delta u_{\varepsilon}\geq-4+\sum\limits_{j=1}^{N_1}\frac{4\varepsilon}{\left(\varepsilon+\abs{x-p_j}^2\right)^2}\geq 0&\text{in }\Omega\\
\Delta v_{\varepsilon}\geq -4+\sum\limits_{j=1}^{N_2}\frac{4\varepsilon}{\left(\varepsilon+\abs{x-q_j}^2\right)^2}\geq 0&\text{in }\Omega\\
u_{\varepsilon}=v_{\varepsilon}=-\ln 2\qquad \text{on }\partial\Omega.
\end{cases}
\end{equation}
By the maximum principle, we conclude that
\begin{equation}\label{maxepsilon}
u_{\varepsilon}<-\ln2\qquad\text{and}\qquad v_{\varepsilon}<-\ln 2\qquad \text{in }\Omega.
\end{equation}
\end{proof}

\textbf{Remark:} Up until Proposition \ref{epsilonbound}, the existence theory holds for $\abs{K}<-4$ as well.  As a matter of fact, we see that when $\abs{K}<-4$, we are able to obtain existence of solutions to the regularized system on a bounded domain.  However, it is here that we are unable to say more when $\abs{K}<-4$ as we are not able to obtain any estimates for $u_{\varepsilon}$ or $v_{\varepsilon}$ which will impede our ability to take a limit as $\varepsilon\rightarrow 0$.  Therefore, in the rest of this section, we restrict ourselves to studying only $-4\leq\abs{K}<0$.

\subsection{Return to original system}
In order to establish the solutions to the original system (\ref{maineq}), prior to regularization, we need to take the limit as $\varepsilon\rightarrow 0$ for $u_{\varepsilon}$ and $v_{\varepsilon}$.  To this end, we add the first two equations in (\ref{s2'}) to get,
\begin{equation}\label{s2in}
\Delta(u_{\varepsilon}+v_{\varepsilon})\geq8\left(e^{u_{\varepsilon}}+e^{v_{\varepsilon}}\right)-8+\sum\limits_{j=1}^{N_1}\frac{4\varepsilon}{\left(\varepsilon+\abs{x-p_j}^2\right)^2}+\sum\limits_{j=1}^{N_2}\frac{4\varepsilon}{\left(\varepsilon+\abs{x-q_j}^2\right)^2},
\end{equation}
and recall the basic inequality
\begin{equation*}
a^2+b^2\geq2ab,\text{ for }a,b>0.
\end{equation*}
We let $a=e^{u_{\varepsilon}/2},b=e^{v_{\varepsilon}/2}$ to see that
\begin{equation}\label{ineq1}
e^{u_{\varepsilon}}+e^{v_{\varepsilon}}\geq 2e^{(u_{\varepsilon}+v_{\varepsilon})/2}.
\end{equation}
In view of (\ref{ineq1}), we are able to find the following estimate for (\ref{s2in}),
\begin{equation}
\Delta(u_{\varepsilon}+v_{\varepsilon})\geq16e^{\frac{1}{2}(u_{\varepsilon}+v_{\varepsilon})}-8+\sum\limits_{j=1}^{N_1}\frac{4\varepsilon}{\left(\varepsilon+\abs{x-p_j}^2\right)^2}+\sum\limits_{j=1}^{N_2}\frac{4\varepsilon}{\left(\varepsilon+\abs{x-q_j}^2\right)^2}.
\end{equation}
Now we see that the average of $u_{\varepsilon}$ and $v_{\varepsilon}$ is a subsolution of
\begin{equation}\label{wsys}
\begin{cases}
\Delta w_{\varepsilon}=8e^{w_{\varepsilon}}-4+\sum\limits_{j=1}^{N_1}\frac{2\varepsilon}{\left(\varepsilon+\abs{x-p_j}^2\right)^2}+\sum\limits_{j=1}^{N_2}\frac{2\varepsilon}{\left(\varepsilon+\abs{x-q_j}^2\right)^2}\qquad\text{in }\Omega\\
w_{\varepsilon}=-\ln2\qquad\text{on }\partial\Omega.
\end{cases}
\end{equation}

Again, by the maximum principle, we see that $w_{\varepsilon}<-\ln2$ in $\Omega$.  Furthermore, since $\frac{1}{2}(u_{\varepsilon}+v_{\varepsilon})$ is a subsolution of (\ref{wsys}), we can see that $w_{\varepsilon}^+=-\ln2$ is a supersolution.  Therefore, by the method of super and subsolutions, we can create a monotone increasing iterative scheme from $\frac{1}{2}(u_{\varepsilon}+v_{\varepsilon})$ to get a solution of (\ref{wsys}).  We note here that,
\begin{equation}\label{wineq}
w_{\varepsilon}\geq\frac{1}{2}(u_{\varepsilon}+v_{\varepsilon}) \qquad\text{in }\Omega,
\end{equation}
and now let
\begin{equation}\label{wzero}
w_0^{\varepsilon}=u_0^{\varepsilon}+v_0^{\varepsilon}+W_0^{\varepsilon},
\end{equation}
where we have defined $u_0^{\varepsilon}$ and $v_0^{\varepsilon}$ in (\ref{backuv}).  As in (\ref{2harm}), with $U_0^{\varepsilon}$ and $V_0^{\varepsilon}$, we let $W_0^{\varepsilon}$ be a harmonic function chosen so that $w_0^{\varepsilon}=-\ln 2$ on $\partial\Omega$.  We will use this construction to show that $\left\{u_{\varepsilon}\right\}$ and $\left\{v_{\varepsilon}\right\}$ are uniformly bounded in $L^2(\Omega)$.  We now let
\begin{equation}\label{wsub}
w_{\varepsilon}=w_0^{\varepsilon}+\tilde{w}_{\varepsilon}.
\end{equation}
So, in view of (\ref{wsub}), we are able to rewrite equation (\ref{wsys}) as,
\begin{equation}\label{wsys2a}
\begin{cases}
\Delta \tilde{w}_{\varepsilon}=8e^{w_0^{\varepsilon}+\tilde{w}_{\varepsilon}}-4+\frac{1}{2}\left(h_1+h_2\right)\qquad\text{in }\Omega\\
\tilde{w}_{\varepsilon}=0\qquad\text{on }\partial\Omega.
\end{cases}
\end{equation}

Since $w_{\varepsilon}<0$, we see that $w_0^{\varepsilon}+\tilde{w}_{\varepsilon}<0$.  Now we multiply the equation in (\ref{wsys2a}) by $\tilde{w}_{\varepsilon}$ and integrate both sides over $\Omega$ to obtain
\begin{equation}\label{wint}
-\int\displaylimits_{\Omega}\abs{\nabla \tilde{w}_{\varepsilon}}^2dx=\int\displaylimits_{\Omega}8e^{w_0^{\varepsilon}+\tilde{w}_{\varepsilon}}\tilde{w}_{\varepsilon}-4\tilde{w}_{\varepsilon}+\frac{1}{2}\left(h_1+h_2\right)\tilde{w}_{\varepsilon}dx.
\end{equation}
We rewrite (\ref{wint}) as
\begin{equation}\label{wtild1}
\|\nabla\tilde{w}_{\varepsilon}\|^2_{L^2(\Omega)}=-8\int\displaylimits_{\Omega}e^{w_0^{\varepsilon}+\tilde{w}_{\varepsilon}}\tilde{w}_{\varepsilon}dx+4\int\displaylimits_{\Omega}\tilde{w}_{\varepsilon}dx-\frac{1}{2}\int\displaylimits_{\Omega}(h_1+h_2)\tilde{w}_{\varepsilon}dx.
\end{equation}
We will use Cauchy's inequality with $\epsilon>0$ as well as the Poincar\'e inequality in (\ref{wtild1}) to obtain,
\begin{equation}\label{wtild}
\|\nabla\tilde{w}_{\varepsilon}\|^2_{L^2(\Omega)}\leq-8\int\displaylimits_{\Omega}e^{w_0^{\varepsilon}+\tilde{w}_{\varepsilon}}\tilde{w}_{\varepsilon}dx+C_1\epsilon^{-1}+{C\epsilon}\|\nabla\tilde{w}_{\varepsilon}\|_{L^2(\Omega)}^2.
\end{equation}

In order to find an upper bound for the first term on the right hand side of (\ref{wtild}), we recall that $w_0^{\varepsilon}+\tilde{w}_{\varepsilon}<0$.  Therefore, using Cauchy's inequality and the Poincar\'e inequality we obtain,
\begin{equation}
\abs{\int\displaylimits_{\Omega}e^{w_0^{\varepsilon}+\tilde{w}_{\varepsilon}}\tilde{w}_{\varepsilon}dx}\leq\int\displaylimits_{\Omega}\abs{\tilde{w}_{\varepsilon}}dx\leq4\abs{\Omega}+\frac{1}{16}\|\tilde{w}_{\varepsilon}\|_{L^2(\Omega)}^2\leq C_2+\frac{1}{2}C\|\nabla\tilde{w}_{\varepsilon}\|_{L^2(\Omega)}^2.
\end{equation}

Using this bound, and inserting into (\ref{wtild}), we see that we have obtained the bound for $\nabla \tilde{w}_{\varepsilon}$ in $L^2(\Omega)$. Moreover, by the Poincar\'e inequality, we have obtained the bound for $\tilde{w}_{\varepsilon}$.  We note that this bound is independent of $\varepsilon>0$.  Furthermore, by (\ref{maxepsilon}), (\ref{wineq}), (\ref{wzero}), and (\ref{wtild}) we see that $\left\{u_{\varepsilon}\right\}$ and $\left\{v_{\varepsilon}\right\}$ are bounded uniformly.

We now use the interior elliptic estimates described in \cite{gilbarg2015elliptic} to assume that there are functions $u$ and $v$
\begin{equation*}
u,v\in C^0\left(\Omega \setminus \left\{p_1,\ldots,p_{N_1},q_1,\ldots,q_{N_2}\right\}\bigcap L^2(\Omega)\right),
\end{equation*}
such that for any compact subset $E\subset\Omega \setminus \left\{p_1,\ldots,p_{N_1},q_1,\ldots,q_{N_2}\right\}$, we have
\begin{equation}\label{convc}
(u_{\varepsilon},v_{\varepsilon})\rightarrow(u,v)\text{in }C^0(E)\text{ as }\varepsilon\rightarrow 0
\end{equation}
and
\begin{equation}\label{convl}
(u_{\varepsilon},v_{\varepsilon})\rightharpoonup(u,v)\text{ weakly in }L^2(\Omega)\text{ as }\varepsilon\rightarrow 0.
\end{equation}
We then use the Green's function to represent the equations in the regularized form (\ref{s2}) with $u=u_{\varepsilon}$ and $v=v_{\varepsilon}$.  We apply the convergence derived in (\ref{convc}) and (\ref{convl}) to see that as $\varepsilon\rightarrow0$, the pair $(u,v)$ satisfies the original equations (\ref{s1}).  In the next section, we take the limit of the bounded domain, $\Omega$, to all of $\mathbb{R}^2$.


\subsection{Single equation limit as $\Omega\rightarrow\mathbb{R}^2$}
In this section, we first consider the single vortex equation subject to the boundary condition,
\begin{equation}
\begin{cases}
\Delta u=8(e^u-1)+4\pi\sum\limits\limits_{j=1}^{N_+}\delta_{\hat{p}_j}\qquad\text{in }B_R\\
u=0\qquad\text{on }\partial B_R. 
\end{cases}
\end{equation}
The ball, $B_R$, is chosen so that it contains all of the vortex points $\hat{p}_j=p_1,\ldots,p_{N_1},q_1,\ldots,q_{N_2}$ and the total vortex number is given by $N_+=N_1+N_2$.  We denote this by
\begin{equation*}
B_R=\left\{x\in\mathbb{R}^2:\abs{x}<R\right\}\quad\text{ and }\quad R>R_0=\max_{1\leq j\leq N_+}\left\{\abs{\hat{p}_j}\right\}.
\end{equation*}
We recall some well known results \cite{lin2007system} related to the equation,
\begin{equation}\label{cara1}
\begin{cases}
-\Delta u+g(x,u)=\mu&\text{in }\Omega\\
u=h\qquad\text{on }\partial\Omega
\end{cases}
\end{equation}
where $\Omega\subset\mathbb{R}^n,n\geq 2$, is a smooth bounded domain, $\mu\in\mathcal{M}(\Omega)$ and $g:\Omega\times\mathbb{R}\rightarrow\mathbb{R}$ is a Carath\'eodory function.  Also, $h\in L^1(\partial\Omega)$.  We recall that $\mathcal{M}$ is the space of finite Radon measures on on open sets $\omega\subset\mathbb{R}^2$.  We equip $\mathcal{M}(\omega)$ with the standard norm,
\begin{equation*}
\|\mu\|_{\mathcal{M}}=\abs{\mu}(\omega)=\int\displaylimits_{\omega}d\abs{\mu}.
\end{equation*}

In the work that follows, we require the use of three necessary results from appendix A of \cite{lin2007system}.  We state them here.  The reader is recommended to see \cite{lin2007system} for proofs of the following.
\begin{thm}\label{a11}
Let $u_1,u_2\in L^1(\Omega)$ be a sub and supersolution of (\ref{cara1}) such that $u_1\leq u_2$ a.e. and $\left[g(\cdot,v)]\cdot\text{dist}(x,\partial\Omega)\right]\in L^1(\Omega)$ for every $v\in L^1(\Omega)$ such that $u_1\leq v\leq u_2$ a.e.  Then (\ref{cara1}) has a solution $u$ such that $u_1\leq u\leq u_2$ a.e..
\end{thm}
\begin{prop}\label{a3}
Suppose that $g(x,u)$ is a Carath\'eodory function that satisfies
\begin{enumerate}
\item $g(x,\cdot)$ is nondecreasing for a.e. $x\in\Omega$;
\item $g(x,0)=0$ for a.e. $x\in\Omega$.
\end{enumerate}
Let $u$ be a solution of (\ref{cara1}).  Then,
\begin{equation}
\int\displaylimits_{\Omega}\abs{g(x,u)}dx\leq\|\mu\|_{\mathcal{M}}\qquad\text{ and}\qquad \int\displaylimits_{\Omega}\abs{\Delta u}\leq 2\|\mu\|_{\mathcal{M}}
\end{equation}
In particular, $g(\cdot,u)\in L^1(\Omega)$ and $\Delta u\in\mathcal{M}(\Omega)$.
\end{prop}

We now state our first result of this section, which is motivated by method of super and subsolutions to the simplified equation.
\begin{prop}
The equation,
\begin{equation}\label{b4}
\begin{cases}
-\Delta u+8(e^u-1)=-4\pi\sum\limits_{j=1}^{N_+}\delta_{\hat{p}_j}\qquad\text{in }B_R\\
u=0\qquad\text{on }\partial B_R.
\end{cases}
\end{equation}
has a solution $u\in L^1(\Omega).$
\end{prop}

\begin{proof}

First, we set $g(x,t)=8(e^t-1)$ and then write 
\begin{equation*}
\mu=-4\pi\sum\limits_{j=1}^{N_+}\delta_{p_j},
\end{equation*} 
where $\mu$ is a finite Radon measure on $\Omega$.
Here, we make use of both Theorem 2 and Prop A.1 in Vazquez \cite{vazquez1983semilinear}, where the equation 
\begin{equation*}
-\Delta u+\beta(u)=f
\end{equation*}
was studied in detail over $\mathbb{R}^2$ with $f$ being a bounded Radon measure, and $\beta$ being a continuous increasing real function with $\beta(0)=0$.  We then find, with data $\mu^{+}$, a supersolution, $U\in L^1(\mathbb{R}^2)$ such that $U\geq 0$ a.e.. 

Furthermore, this $U$ satisfies $e^U-1\in L^1(\Omega).$  As in appendix B of \cite{lin2007system}, we let $v\in L^1(\Omega)$ be the solution of
\begin{equation}
\begin{cases}
-\Delta v=-\mu^-\qquad\text{in }\Omega\\
v=0\qquad\text{on }\partial\Omega.
\end{cases}
\end{equation}

Then $U$ and $v$ are super and subsolutions of (\ref{b4}) respectively.  By Theorem \ref{a11} we see that (\ref{b4}) has a solution $u\in L^1(\Omega)$.


\end{proof}

The next two results are special cases of the work done by Lin, Yang, and Ponce \cite{lin2007system}.  The first result establishes the radial behavior of solutions to (\ref{b4}). 
\begin{lem}\label{51}
Every solution of (\ref{b4}) increases along any radial direction on $\mathbb{R}^2\setminus B_{R_0}$.
\end{lem}
\begin{proof}
We first let $A_R=\left\{x| R_0\leq\abs{x}\leq R\right\}$ and see that on $A_R$, the equation (\ref{b4}) becomes
\begin{equation}\label{b44}
\begin{cases}
-\Delta u+8(e^u-1)=0\qquad\text{in }A_R\\
u=0\qquad\text{ when }\abs{x}=R.
\end{cases}
\end{equation}
If we let $c(x)=e^{u(x)+\xi(x)}$, where $u(x)+\xi(x)=\ln\left(\frac{1-e^u}{u}\right)$, and recall that by the maximum principle, $u<0$ in $B_R$, we are able to use the Hopf Boundary Lemma on $\Delta u=8c(x)u$ to determine that
\begin{equation}
\frac{\partial u}{\partial\text{n}}>0\quad\text{ if }\abs{x}=R.
\end{equation}
In other words, we see that $u$ is increasing at the boundary of $B_R$.  We want to show that $u$ is increasing everywhere outside of $B_R$.  Furthermore,  Aleksandrov, Gidas, Ni, and Nirenberg \cite{gidas1979symmetry} established that a function $u$ satisfying $\Delta u=8c(x)u$ with $u<0$ inside $B_R$ and $u=0$ on $\partial B_R$ will be radially symmetric outside of $B_R$ via the method of moving planes.  This means that without loss of generality, we may show that $u$ is increasing in the $x_1$ direction.

We then let $\sigma\in(R_0,R)$ and define the half plane,
\begin{equation}
\Sigma_{\sigma}=\left\{x\in B_R|x_1>\sigma\right\}.
\end{equation}
We consider the reflection of $x$ over the line $x_1=\sigma$ and denote this by $x^{\sigma}=(2\sigma-x_1,x_2)$.  We then set $u_{\sigma}=u(x^{\sigma})$ for $x\in\Sigma$.  We note here that upon setting $w_{\sigma}=u(x)-u_{\sigma}(x)$ for $x\in\Sigma_{\sigma}$, the proof becomes a special case of that in Lemma 5.1 of \cite{lin2007system} and the result immediately follows.
\end{proof}

We now consider (\ref{b4}) over the full space, and multiply by a negative to rewrite the equation as
\begin{align}\label{1eqr}
\Delta u+8(1-e^u)=4\pi\sum\limits_{j=1}^{N_+}\delta_{p_j}\qquad\text{in }\mathbb{R}^2,
\end{align}
and are now able to establish the following asymptotic estimate of topological solutions to (\ref{1eqr}).

\begin{lem}\label{52}
Let $u$ be a solution to (\ref{1eqr}) such that
\begin{equation}
	u(x)<0\qquad\text{ and }\qquad \int\displaylimits_{\mathbb{R}^2}\left(1-e^u\right)dx<\infty.
\end{equation}
Then we have the asymptotic estimate
\begin{equation}\label{131}
u(x)=-\lambda \ln\abs{x}+O(1)
\end{equation}
where
\begin{equation}\label{132}
\lambda=\frac{4}{\pi}\int\displaylimits_{\mathbb{R}^2}\left(1-e^u\right)dx-2N_+.  
\end{equation}
The solutions are topological if and only if $\lambda=0$.  That is, they satisfy $u(x)\rightarrow0$ as $\abs{x}\rightarrow\infty$. Otherwise, the solutions are nontopological and satisfy $u(x)\rightarrow -\infty$ as $\abs{x}\rightarrow\infty$.  
\end{lem}
The proof of this lemma is only a modification of Lemma 5.2 in \cite{lin2007system} and Theorem 1.1 in \cite{cheng1997asymptotic} with the term $\mathcal{K}$ being taken as $\mathcal{K}(x)=8(e^{-u}-1)$.  Then we have the modification,
\begin{equation*}
c_1e^{-\abs{x}^{\beta}}\leq \mathcal{K}(x)\leq C_2\abs{x}^m.
\end{equation*}
Since the proof would just follow that of \cite{cheng1997asymptotic,lin2007system}, the details are omitted here.

In order to extend the domain to the full-plane,  we consider the system in terms of a sequence, $\left\{u_n\right\}_{n=1}^{\infty}$, that satisfies,
\begin{align}\label{1eqn}
\begin{cases}
\Delta u_n=8(e^{u_n}-1)+4\pi\sum\limits_{j=1}^{N_+}\delta_{\hat{p}_j}\qquad\text{in }\Omega_n\\
u_n=0\qquad\text{on }\partial \Omega_n.
\end{cases}
\end{align}
where we consider the domain $\Omega_n=B_{R_n}$ and $R_n\rightarrow\infty$.

The main result of this section is stated below.  

\begin{lem}\label{53}
Let $\left\{u_n\right\}_{n=1}^{\infty}$ be a sequence of solutions to (\ref{1eqn}).  Then, there is a subsequence, $\left\{u_{n_m}\right\}_{m=1}^{\infty}$, which converges pointwise to a solution, $u$, of (\ref{1eqr}) satisfying the topological boundary condition,
\begin{equation}
\lim\limits_{\abs{x}\rightarrow\infty}u(x)=0.
\end{equation}
\end{lem}

\begin{proof}
To begin, we define a background function $w_0$ as in (\ref{backuv}) and
see that $w_0$ is bounded for almost every $x\in\mathbb{R}^2$.  We then use this background function to express the sequence of functions $u_n$ as 
\begin{equation}
u_n=w_0+v_n,
\end{equation}
and proceed to prove, by contradiction, that the sequence of functions $\left\{v_n\right\}_{n=1}^{\infty}$ is uniformly bounded.  As we will see, the contradiction will arise as a result of the nontopological boundary conditions stated in Lemma \ref{52}.    

Now for the assumption.  To this end, we suppose that $\left\{v_n\right\}_{n=1}^{\infty}$ is an unbounded sequence of functions and we recall that $u_n<0$.  Since $v_n=u_n-w_0$, we see that by assuming the sequence of functions, $\left\{v_n(x)\right\}$ is unbounded, we mean $v_n\rightarrow-\infty$ and we may assume it does so uniformly. By this assumption of unboundedness, we may further assume that there exists a sequence $\left\{x_n\right\}_{n=1}^{\infty}$ in $\mathbb{R}^2$ so that
\begin{equation}\label{dun}
v_n(x_n)\rightarrow-\infty \qquad\text{ as }\qquad n\rightarrow\infty.
\end{equation}
We now make the following claim.  

\textit{Claim 1.} There exists a sequence $\left\{x_n\right\}_{n=1}^{\infty}$ in $\mathbb{R}^2$ such that for some constant $c>0$,
\begin{equation}
u_n(x_n)=-c\qquad\text{ and }\qquad\text{dist}(x_n,\partial B_{R_n})\rightarrow\infty.
\end{equation}
Prior to establishing claim 1, we denote by $\bar{u}_n(r)$, the average value of $u_n$ on $B_r$,
\begin{equation}
\bar{u}_n(r)=\frac{1}{2\pi}\int\displaylimits_0^{2\pi}u_n(re^{i\theta})d\theta.
\end{equation}
Assume that claim 1 is false.  That is, we assume that there exists a $\kappa>0$ such that if $u_n(x)\geq-c$ then $\text{dist}(x,\partial B_{R_n})\leq\kappa$.  Recall that $R_0$ is the smallest radius such that $B_{R_0}$ contains all the point vortices and take $n$ large enough to assume that $R_n\geq R_0+\kappa$, or $B_{R_0+\kappa}\subset B_{R_n}$.

It is clear that since $\Omega_n=B_{R_n}$, and $u_n=0$ on $\partial\Omega_n$,
\begin{equation}\label{mvt1}
\bar{u}_n(R_n-\kappa)\leq -c\qquad\text{and}\qquad\bar{u}_n(R_n)=0.
\end{equation}
Then, by the mean value theorem and (\ref{mvt1}), there exists an $r_n$ such that
\begin{equation}
R_n-\kappa<r_n<R_n
\end{equation}
and
\begin{equation}
\bar{u}_n'(r_n)=\frac{\bar{u}_n(R_n)-\bar{u}_n(R_n-\kappa)}{R_n-(R_n-\kappa)}\geq\frac{c}{\kappa}.
\end{equation}
We then recall that,
\begin{equation}\label{useful1}
\bar{u}_n'(r)=\frac{1}{2\pi r}\int\displaylimits_{B_r}\Delta u_ndx\qquad\forall r>R_0,
\end{equation}
and see by setting $r=r_n$,
\begin{equation}\label{barn}
\bar{u}_n'(r_n)=\frac{1}{2\pi r_n}\int\displaylimits_{B_{r_n}}\Delta u_ndx.
\end{equation}
We multiply (\ref{barn}) by $r_n$ to get,
\begin{align}\label{contradict1}
\frac{r_n\cdot c}{\kappa}\leq r_n\bar{u}_n(r_n)=\frac{4}{\pi}\int\displaylimits_{B_{R_n}}(e^{u_n}-1)dx+2N_+\leq 2N_+.
\end{align}
We now take $n\rightarrow\infty$ in (\ref{contradict1}) and see that since $N_+$ is a finite number and $r_n\rightarrow\infty$ as $n\rightarrow\infty$ the following provides us with a contradiction
\begin{equation}
\frac{r_n\cdot c}{\kappa}\leq 2N_+.
\end{equation}
Therefore, the assumption is indeed false and we have established claim 1.  

Now we recall that as $n\rightarrow\infty$,  the functions $u_n\rightarrow-\infty$ uniformly on any compact subset of $\mathbb{R}^2$.  Therefore, the sequence defined in (\ref{dun}) satisfies $\abs{x_n}\rightarrow\infty$ and we set 
\begin{equation}\label{bigu}
U_n(x)=u_n(x+x_n)
\end{equation}
and make our second claim.\\

\textit{Claim 2}: The sequence of functions $U_n$ defined in (\ref{bigu}) converge to a solution, $U$ of the equation,
\begin{equation}\label{fullU}
\begin{cases}
\Delta U=8(e^{U}-1),\qquad\text{in }\mathbb{R}^2\\
U\leq 0\qquad\text{and}\qquad U(0)=-c.
\end{cases}
\end{equation}
To prove this claim, we see that 
\begin{equation*}
U_n(0)=u_n(x_n)=-c,
\end{equation*}
where $U_n$ is defined in the ball $B_{\rho_n}$ where $\rho_n\rightarrow\infty$ and $B_{\rho_n}$ does not contain any $\hat{p}_j$, $j=1,\ldots,N_+$.  This translation of $u_n$ yields the following differential equation without any $\delta_{\hat{p}_j}$ terms,
\begin{equation}\label{singleU}
\begin{cases}
\Delta U_n=8(e^{U_n}-1),\qquad\text{in }B_{\rho_n}\\
U_n\leq 0\qquad\text{and}\qquad U_n(0)=-c.
\end{cases}
\end{equation}
By Lemma \ref{51}, we may assume that $U_n(x)$ increases along the positive $x_2$-axis. We then use (\ref{useful1}) and integrate the single equation in (\ref{singleU}) to obtain,
\begin{equation}\label{fromthis}
r\bar{U}_n'(r)=\frac{4}{\pi}\int\displaylimits_{B_r}(e^{U_n}-1)dx.
\end{equation}
From (\ref{fromthis}), along with the fact that $U_n\leq 0$, we obtain,
\begin{equation}
\abs{\bar{U}_n'(r)}\leq\frac{4}{\pi r}\int\displaylimits_{B_r}\abs{e^{U_n}-1}dx\leq 4r.
\end{equation}
We then use the average value of $U_n$,
\begin{equation}
\abs{\bar{U}_n(r)}=\frac{1}{2\pi}\int\displaylimits_0^{2\pi}\abs{U_n(re^{i\theta})}d\theta,
\end{equation}
along with the mean value theorem to obtain the following bound for $\bar{U}_n$,
\begin{equation}
\abs{\bar{U}_n(r)}\leq \frac{1}{2}+4r^2.
\end{equation}
We see that $\left\{U_n\right\}_{n=1}^{\infty}$ has a uniform $L^1$ bound on $\partial B_r$.  We then see that by elliptic estimates \cite{gilbarg2015elliptic}, the sequence $\left\{U_n\right\}$ is uniformly bounded over any compact subset of $\mathbb{R}^2$.  Therefore, we can extract a subsequence (if necessary) so that $U_n\rightarrow U$ a solution of,
\begin{equation}
\begin{cases}
\Delta U=8(e^U-1)\qquad\text{in }\mathbb{R}^2\\
U\leq 0\qquad U(0)=-c.
\end{cases}
\end{equation}
We have thus established claim 2.

Before proceeding, we recall that $u_n$ satisfies system in $\Omega_n$ (\ref{1eqn}). Since $u_n$ increases in any radial direction, we see that $\frac{\partial u_n}{\partial\text{n}}>0$ on $\partial B_{R_n}$.  With this in mind, we let $\hat{w}_0=\ln\abs{x-\hat{p}_j}^2$ when $x$ is in a small neighborhood of each $\hat{p}_j$.  In this construction of $\hat{w}_0$, we create a function with support in $B_{R_0}$ and is smooth away from the point vortices $\hat{p}_j$.  We now make the decomposition,
\begin{equation}
u_n=\hat{w}_0+V_n,
\end{equation}
to yield
\begin{equation}\label{VN}
\begin{cases}
\Delta V_n=8(e^{V_n+\hat{w}_0}-1)+g(x)\qquad\text{in }B_{R_n}\\
V_n=0\qquad \frac{\partial V_n}{\partial\text{n}}>0\quad\text{on }\partial B_{R_n}.
\end{cases}
\end{equation}
We then use integration by parts in (\ref{VN}) to obtain
\begin{equation}\label{VN2}
0<\int\displaylimits_{\partial B_{R_n}}\frac{\partial V_n}{\partial\text{n}}dS=8\int\displaylimits_{B_{R_n}}(e^{V_n+\hat{w}_0}-1)dx+\int\displaylimits_{B_{R_n}}g(x)dx.
\end{equation}
From (\ref{VN2}), it is clear that 
\begin{equation}\label{uniformbound}
8\int\displaylimits_{B_{R_n}}\bigg(1-e^{u_n}\bigg)dx\leq\int\displaylimits_{B_{R_n}}\abs{g(x)}dx\leq\int\displaylimits_{\mathbb{R}^2}\abs{g(x)}dx
\end{equation}
where the fixed function, $g(x)$, satisfies
\begin{equation}
\int\displaylimits_{\mathbb{R}^2}\abs{g(x)}dx<\infty.
\end{equation}
We then replace $u_n$ with $U_n$ in (\ref{uniformbound}) and let $n\rightarrow\infty$ to obtain

\begin{equation}
8\int\displaylimits_{B_{R_n}}\bigg(1-e^{U_n}\bigg)dx<\infty.
\end{equation}
We are then able to call upon Lemma \ref{52} and see that
\begin{equation}
\lim\limits_{\abs{x}\rightarrow\infty}\frac{U(x)}{\ln\abs{x}}=-\lambda,\qquad \lambda=\frac{4}{\pi}\int\displaylimits_{\mathbb{R}^2}(1-e^U)dx>0.
\end{equation}
We recall that $U_n$ is nondecreasing and thus $U$ is nondecreasing. Since $\lambda \neq 0$, we must have $\displaystyle\lim\limits_{\abs{x}\rightarrow\infty}U(x)=-\infty$ from Lemma \ref{52}.  This is a contradiction since $U_n$ and thus $U$ is nondecreasing and $U(0)=-c$.  Therefore, the assumption that $v_n$ (and consequently  $u_n$) being uniformly unbounded is false.  We may now extract a subsequence of $\left\{u_n\right\}$ converging to $u$ and satisfies $u\rightarrow 0$ as $\abs{x}\rightarrow\infty$. 
\end{proof}

\subsection{Variational solutions of the system in the limit $\Omega\rightarrow\mathbb{R}^2$}

In this section, we take the domain limit of the solutions obtained in section 2.2 in order to establish existence of solutions over the entire plane.  We state the main result of this section as the following theorem.

\begin{thm}\label{limitboundary}
On the full-plane, $\Omega=\mathbb{R}^2$, the system (\ref{maineq}) has a solution pair $(u,v)$ satisfying the boundary condition
\begin{equation}\label{boundcondition}
u,v\rightarrow -\ln 2\qquad\text{as}\qquad\abs{x}\rightarrow\infty.
\end{equation}
\end{thm}  
\begin{proof}
We accomplish this by taking a sequence of open balls containing all the $p_j$ and $q_j$ terms.  To this end, we let $\left\{R_n\right\}_{n=1}^{\infty}$ be a sequence such that,
\begin{equation*}
R_n>\max\left\{\abs{p_j},\abs{q_{j'}}\right\}
\end{equation*}
where $j=1,\ldots,N_1$ and $j'=1,\ldots,N_2$.  We also require the sequence to satisfy
\begin{equation*}
R_n\rightarrow\infty \text{ as }n\rightarrow\infty.
\end{equation*}
By the work done in the previous sections, we see that if $\Omega=B_{R_n}$, then the system (\ref{s1}) has a solution.  We will denote this solution by $(u_n,v_n)$.  To make the boundary conditions homogenous, we consider the shifts 
\begin{equation}
u_n=\tilde{u}_n+\ln2\qquad\text{and}\qquad v_n=\tilde{v}_n+\ln2,
\end{equation}
Now we have $\tilde{u}_n=\tilde{v}_n=0$ on $\partial B_{R_n}$ and the following system
\begin{equation}\label{f1'}
\begin{cases}
\Delta \tilde{u}_n=2k_{11}e^{\tilde{u_n}}+2k_{12}e^{\tilde{v}_n}-4+4\pi\sum\limits_{j=1}^{N_1}\delta p_{j}&\text{in }B_{R_n}\\
\Delta \tilde{v}_n=2k_{12}e^{\tilde{u}_n}+2k_{11}e^{\tilde{v}_n}-4+4\pi\sum\limits_{j=1}^{N_2}\delta q_{j}&\text{in }B_{R_n}\\
\tilde{u}_n=\tilde{v}_n=0\qquad\text{on }\partial B_{R_n}.
\end{cases}
\end{equation}
We add the two equations in (\ref{f1'}) together and see that the average of $\tilde{u}_n$ and $\tilde{v}_n$ is a subsolution of
\begin{equation}\label{wsys2}
\begin{cases}
\Delta w=8(e^{w}-1)+4\pi\sum\limits_{j=1}^{N_+}\delta_{\hat{p}_j}\\
w=0\qquad\text{on }\partial B_{R_n},
\end{cases}
\end{equation}
where we have simplified the summations by writing
\begin{equation}
\sum\limits_{j=1}^{N_1}\delta_{p_j}+\sum\limits_{j=1}^{N_2}\delta_{q_j}=\sum\limits_{j=1}^{N_+}\delta_{\hat{p}_j}.
\end{equation}

As in \cite{sy3}, we begin with the subsolution, $\frac{1}{2}(\tilde{u}_n+\tilde{v}_n)$, and supersolution, $0$.  Then, we iterate in a monotone increasing fashion to obtain a solution of (\ref{wsys2}).  In this situation, we have,
\begin{equation}\label{wuv}
0>w_n\geq \frac{1}{2}(\tilde{u}_n+\tilde{v}_n),\text{in }B_{R_n}.
\end{equation}
Here, we use Lemma \ref{53} and assume that the sequence $\left\{w_n\right\}_{n=1}^{\infty}$ converges pointwise to $w$ which is a solution of,
\begin{equation}\label{wsys3}
\begin{cases}
\Delta w=8(e^{w}-1)+4\pi\sum\limits_{j+1}^{N_+}\delta_{\hat{p}_j}\\
\lim\limits_{\abs{x}\rightarrow\infty}w(x)=0,\qquad w<0\text{ a.e. }.
\end{cases}
\end{equation}

In order to establish the limit for $\tilde{u}_n$ and $\tilde{v}_n$, we take another subsequence (if necessary) and then for some $\tilde{u}$ and $\tilde{v}$, we have
\begin{equation}
u_n\rightarrow \tilde{u}\qquad\text{and}\qquad v_n\rightarrow \tilde{v}.
\end{equation}
This convergence is pointwise on $\mathbb{R}^2$.  We see that both of the functions $\tilde{u}$ and $\tilde{v}$ are negative and by taking the limit in (\ref{wuv}), we obtain,
\begin{equation}
0>w\geq\frac{1}{2}(\tilde{u}+\tilde{v}).
\end{equation}
Furthermore, both $\tilde{u}_n$ and $\tilde{v}_n$ are bounded in $\mathbb{R}^2$.  We write this as
\begin{equation}
\abs{\tilde{u}_n}\leq M_u,\qquad \abs{\tilde{v}_n}\leq M_v,\qquad\text{ a.e.}.
\end{equation}
Then, we can let $M=\max\left\{M_{\tilde{u}},M_{\tilde{v}}\right\}$ and see that,
\begin{equation}
-\frac{1}{2}M\leq \frac{1}{2}(\tilde{u}_n+\tilde{v}_n)<w_n<0.
\end{equation}
Moreover, we see that if we multiply $w_n$ by a large enough constant, $M'$, we obtain,
\begin{equation}
M'w_n\leq \tilde{u}_n<0,\qquad M'w_n\leq \tilde{v}_n< 0.
\end{equation}
Then, by taking the limit as $\abs{x}\rightarrow\infty$, we achieve,
\begin{equation}
\lim\limits_{\abs{x}\rightarrow\infty}\tilde{u}_n=\lim\limits_{\abs{x}\rightarrow\infty}\tilde{v}_n=0
\end{equation}
which implies that $\tilde{u},\tilde{v}\rightarrow 0$ as $\abs{x}\rightarrow\infty$.  In terms of the functions $u$ and $v$, we have obtained
\begin{equation}
u,v\rightarrow -\ln 2\qquad\text{as}\qquad\abs{x}\rightarrow\infty.
\end{equation}
Therefore, the boundary condition (\ref{boundcondition}) is satisfied on $\mathbb{R}^2$.
\end{proof}

\subsection{Asymptotic analysis}
In this section, we discuss the behavior of solutions $u$ and $v$ to the system (\ref{maineq}) over $\mathbb{R}^2$.  We now prove Theorem \ref{asym}

Before proceeding, we recall a well known proposition presented by Bellman in \cite{bellman2008stability}.
\begin{prop}\label{phiprop}
Let $\alpha,\lambda,t_0>0$ and let $\Phi:[t_0,\infty)\rightarrow \mathbb{R}$ be a continuous function such that
\begin{equation*}
\lim\limits_{t\rightarrow \infty}\Phi(t)=0\qquad\text{and}\qquad
\int\displaylimits_{t_0}^{\infty}\abs{\Phi(t)}dt<\infty.
\end{equation*}
Then the equation
\begin{equation}
\begin{cases}
w''+\frac{1}{t}w'-\left(\lambda+\Phi(t)\right)w=0\qquad\text{in }(t_0,\infty)\\
w(t_0)=\alpha\qquad\text{ and }\qquad\lim\limits_{t\rightarrow\infty}w(t)=0
\end{cases}
\end{equation}
has a unique solution $w_0$.  Furthermore, there exist positive constants $C_0$ and $C_1$ satisfying
\begin{equation}
C_0\leq\frac{w_0(t)}{W_0(t)}\leq C_1\qquad \forall t\geq t_0
\end{equation}
where
\begin{equation}
W_0=\frac{e^{-\sqrt{\lambda}t}}{t^{\frac{1}{2}}}.
\end{equation}
\end{prop}
We are now ready to prove Theorem \ref{asym}.  
\begin{proof}
First, we consider the shift
\begin{equation}\label{shift}
\tilde{u}=u+\ln2\qquad\text{and}\qquad \tilde{v}=v+\ln 2.
\end{equation}
Upon insertion of (\ref{shift}) into (\ref{maineq}) we obtain the following system.
%
\begin{equation}\label{maineq4}
\begin{cases}
\Delta \tilde{u}=2k_{11}e^{\tilde{u}}+2k_{12}e^{\tilde{v}}-4+4\pi\sum\limits_{j=1}^{N_1}\delta_{p_j}(x)&\text{in }\mathbb{R}^2\\
\Delta \tilde{v}=2k_{12}e^{\tilde{v}}+2k_{11}e^{\tilde{v}}-4+4\pi\sum\limits_{j=1}^{N_2}\delta_{q_j}(x)&\text{in }\mathbb{R}^2.\\
\tilde{u},\tilde{v}\rightarrow 0\qquad\text{ as }\abs{x}\rightarrow\infty.
\end{cases}
\end{equation}
We now establish the estimate for $\tilde{u}$ and see that a similar argument will provide the estimate for $\tilde{v}$.  Here, we consider a ball, $B_R$, and $\varepsilon>0$ such that $p_j,q_j\in B_{R-\varepsilon}$
We then define a new set, $A_R=\mathbb{R}^2\setminus\overline{B_{R-\varepsilon}}$, and obtain
\begin{equation}
\Delta \tilde{u}=2k_{11}e^{\tilde{u}}+2k_{12}e^{\tilde{v}}-4\qquad\text{in }A_R.
\end{equation}
We then use Kato's inequality  \cite{kato1972schrodinger},
$\Delta \abs{u}\geq\text{sgn}(u)\Delta u\qquad\text{in }\mathcal{D}'(A_R)$, along with the fact that $k_{12}=2-k_{11}$ to get
\begin{equation}
\Delta\abs{\tilde{u}}\geq\text{sgn}(\tilde{u})\Delta \tilde{u}
=4(1-e^{\tilde{v}})+2k_{11}(e^{\tilde{v}}-1)-2k_{11}(e^{\tilde{u}}-1)
\end{equation}
which can be summarized nicely as
\begin{equation}\label{subharm}
\Delta\abs{\tilde{u}}\geq(4-2k_{11})\abs{e^{\tilde{v}}-1}+2k_{11}\abs{e^{\tilde{u}}-1}\qquad\text{in }\mathcal{D}'(A_R).
\end{equation}
Since  $-4\leq\abs{K}<0$, we see that the right hand side of (\ref{subharm}) is positive.   Therefore, $\tilde{u}$ is subharmonic in $A_R$.  In view of this, and using the fact that $\tilde{u}\in L^1(\mathbb{R}^2)$, we obtain for any given $x\in\mathbb{R}^2\setminus B_R$, 
\begin{equation}\label{subharm2}
\abs{\tilde{u}(x)}\leq\frac{1}{\pi r^2}\int\displaylimits_{B_R(y)}\abs{\tilde{u}}dy,\qquad\text{ for all } 0<r\leq\abs{x}-R+\varepsilon.
\end{equation}
If we take $r=\abs{x}-R+\varepsilon$ in (\ref{subharm2}) we obtain,
\begin{equation}
\nonumber\abs{\tilde{u}}\nonumber\leq\frac{1}{\pi(\abs{x}-R+\varepsilon)^2}\int\displaylimits_{B_{\abs{x}-R+\varepsilon}}\abs{\tilde{u}}dy\nonumber\leq\frac{1}{\pi(\abs{x}-R+\varepsilon)^2}\int\displaylimits_{\mathbb{R}^2}\abs{\tilde{u}}dy\leq\frac{C}{\abs{x}^2}.
\end{equation}
We obtain a similar result for $\tilde{v}$.  Therefore, we have established the estimate,
\begin{equation}\label{uvest}
\abs{u(x)+\ln2}+\abs{v(x)+\ln2}\leq\frac{C}{\abs{x}^2}\qquad \text{ for all }x\in\mathbb{R}^2\setminus B_R.
\end{equation}
Next, we define for every $r\geq R$, 
 \begin{equation}\label{phi1}
\Phi(r)=4\min_{\abs{x}=r}\left\{\frac{k_{12}}{2}\frac{\abs{e^{\tilde{v}}-1}}{\abs{\tilde{v}}}+\frac{k_{11}}{2}\frac{\abs{e^{\tilde{u}}-1}}{\abs{\tilde{u}}}\right\}.
\end{equation}
From (\ref{uvest}), we see that both $\tilde{u}(x)$ and $\tilde{v}(x)$ are uniformly bounded on $\mathbb{R}^2\setminus B_R$.  Therefore, it follows from elliptic estimates  that both $\tilde{u}$ and $\tilde{v}$ are continuous.  This implies that $\Phi$ is continuous.  

We now show that $\Phi\rightarrow0$ as $\abs{x}=r\rightarrow\infty$.  To this end, we observe that for $-L\leq s,t\leq L$, 
\begin{equation}
\abs{\frac{k_{12}}{2}\frac{\abs{e^s-1}} {\abs{s}}+\frac{k_{11}}{2}\frac{\abs{e^t-1}}{\abs{t}}-1}
=\abs{\frac{k_{12}}{2}\left(\frac{\abs{e^s-1}}{\abs{s}}-1\right)+\frac{k_{11}}{2}\left(\frac{\abs{e^t-1}}{\abs{t}}-1\right)}\\
\leq C\left(\abs{s}+\abs{t}\right)
\end{equation}
where we have used the convention ${\abs{e^s-1}}/{s}=1$ if $s=0$ and the identity ${k_{11}}/{2}+{k_{12}}/{2}=1.$ It is now clear that
\begin{equation}\label{phib}
\abs{\Phi(r)}\leq C\max_{\abs{x}=r}\left\{\abs{\tilde{u}(x)}+\abs{\tilde{v}(x}\right\}\leq \frac{C}{r^2}.
\end{equation}
From (\ref{phib}), we see that 
\begin{equation}
\lim\limits_{r\rightarrow\infty}\Phi(r)=0\qquad\text{ and }\qquad\int\displaylimits_R^{\infty}\abs{\Phi(r)}dr<\infty.
\end{equation}

Now, we have a function, $\Phi$, that satisfies the conditions of Proposition \ref{phiprop}. Let $w_0$ be the unique radial solution of 
\begin{equation}
\begin{cases}
-\Delta w_0+\left(4+\Phi(x)\right)w_0=0\qquad\text{in }\mathbb{R}^2\setminus \overline{B_R},\\
w_0=M\qquad\text{on }\partial B_R\\
\lim\limits_{\abs{x}\rightarrow\infty}w_0(x)=0.
\end{cases}
\end{equation}
Where
\begin{equation*}
M=\max\left\{M_u,M_v\right\},\quad
M_{u}=\max_{x\in\mathbb{R}^2\setminus B_R}\abs{\tilde{u}(x)},\quad \text{and}\qquad M_{v}=\max_{x\in\mathbb{R}^2\setminus B_R}\abs{\tilde{v}(x)}
\end{equation*}
and $\Phi$ is as defined in (\ref{phi1}).  We make the important observation that
\begin{equation}\label{uvw}
\abs{\tilde{u}},\abs{\tilde{v}}\leq w_0\qquad\text{in }\mathbb{R}^2\setminus B_R. 
\end{equation}
 Now given $\varepsilon>0$, we take $R'>R$ large enough so that 
\begin{equation}
\abs{\tilde{u}}\leq\varepsilon \qquad\text{ for all }x\in\mathbb{R}^2\setminus B_{R'}.
\end{equation}

We let $Z=\abs{\tilde{u}}-w_0-\varepsilon$ and see that by Kato's inequality \cite{kato1972schrodinger} we have
\begin{equation}\label{230}
\begin{cases}
\Delta Z+(4+\Phi)w_0\geq(4-2k_{11})\abs{e^{\tilde{v}}-1}+2k_{11}\abs{e^{\tilde{u}}-1}\qquad\text{in }B_{R'}\setminus \overline{B_R}\\
Z\leq 0\qquad\text{on }\partial B_R\cup\partial B_{R'}.
\end{cases}
\end{equation}
It will be very useful in what follows for us to write (\ref{230}) as
\begin{equation}
\begin{cases}
-\Delta Z-(4+\Phi)w_0+2k_{12}\abs{e^{\tilde{v}}-1}+2k_{11}\abs{e^{\tilde{u}}-1}\leq 0\qquad\text{in }B_{R'}\setminus \overline{B_R}\\
Z\leq 0\qquad\text{on }\partial B_R\cup\partial B_{R'}.
\end{cases}
\end{equation}

We use the fact that $Z=0$ on the boundary, along with integration by parts, to obtain
\begin{equation}
-\int\displaylimits_{B_{R'}\setminus\overline{B_R}}Z\Delta\eta dx\leq\int\displaylimits_{B_{R'}\setminus\overline{B_R}}\left\{(4+\Phi)w_0-2k_{12}\abs{e^{\tilde{v}}-1}-2k_{11}\abs{e^{\tilde{u}}-1}\right\}\eta dx
\end{equation}
for every $\eta\in C_0^2(\overline{B_{R'}}\setminus B_R)$ with $\eta\geq 0$ in $B_{R'}\setminus\overline{B_R}$. Now by Proposition B.5 in \cite{brezis2007nonlinear}, which is a variant of Kato's inequality for functions without compact support, we see that
\begin{align}
-\int\displaylimits_{B_{R'}\setminus\overline{B_R}} Z^+\Delta\eta dx&\nonumber\leq\int\displaylimits_{[\abs{u}\geq w_0+\varepsilon]}\left\{(4+\Phi)w_0-2k_{12}\abs{e^{\tilde{v}}-1}-2k_{11}\abs{e^{\tilde{u}}-1}\right\}\eta dx\\
&\nonumber\leq\int\displaylimits_{[\abs{u}\geq w_0+\varepsilon]}\left\{(4+\Phi)-2k_{12}\frac{\abs{e^{\tilde{v}}-1}}{\abs{\tilde{v}}}-2k_{11}\frac{\abs{e^{\tilde{u}}-1}}{\abs{\tilde{u}}}\right\}w_0\eta dx\\&\leq 0
\end{align}

We note that $w_0$ and $\eta$ are both nonnegative.  Also, the term inside of the brackets is nonpositive.  Therefore, we conclude that $Z^+=\max(0,Z)\leq 0$ which implies that $Z\leq 0$..  From this we see that
\begin{equation}\label{uw1}
\abs{\tilde{u}}\leq w_0+\varepsilon\qquad\text{in }B_{R'}\setminus\overline{B_R},
\end{equation}
and take the limit as $R'\rightarrow\infty$ to see that $B_{R'}\rightarrow\mathbb{R}^2$. Then (\ref{uw1}) becomes
\begin{equation}\label{uw2}
\abs{\tilde{u}}\leq w_0+\varepsilon\qquad\text{in }\mathbb{R}^2\setminus\overline{B_R}.
\end{equation}
Finally, we use Proposition \ref{phiprop} with $W_0={e^{-2t}}{t^{-\frac{1}{2}}}$ 
along with the bound in (\ref{uvw}) to obtain
\begin{equation}
\abs{\tilde{u}(x)}\leq \frac{C}{2}\frac{e^{-2\abs{x}}}{\abs{x}^{\frac{1}{2}}}\qquad\text{and}\qquad
\abs{\tilde{v}(x)}\leq \frac{C}{2}\frac{e^{-2\abs{x}}}{\abs{x}^{\frac{1}{2}}}.
\end{equation}
Therefore, we have obtained

\begin{equation}
\abs{\tilde{u}(x)}+\abs{\tilde{v}(x)}\leq C\frac{e^{-2\abs{x}}}{\abs{x}^{\frac{1}{2}}},
\end{equation}
which can be written as
\begin{equation}
\abs{u(x)+\ln2}+\abs{v(x)+\ln2}\leq C\frac{e^{-2\abs{x}}}{\abs{x}^{\frac{1}{2}}}.
\end{equation}

We have thus established the first result in Theorem \ref{asym}.  To establish the second result , we use a result developed in \cite{bethuel1993asymptotics,lin2007system}.
\begin{lem}\label{lemlem}
Let $u,f\in L^{\infty}(B_1)$ be such that
\begin{equation*}
-\Delta u=f\qquad\text{in }\mathcal{D}'(B_1).
\end{equation*}
Then,
\begin{equation}
\|\nabla u\|_{L^{\infty}\left(B_{1/2}\right)}^2\leq C(\|u\|_{L^{\infty}(B_1)}+\|f\|_{L^{\infty}(B_1)})\|u\|_{L^{\infty}(B_1)}.
\end{equation}

\end{lem}
We apply Lemma \ref{lemlem} to $\tilde{u}$ and $\tilde{v}$ on balls $B_1(x)$ for $\abs{x}\geq R+1$ to obtain
\begin{equation}
\abs{\nabla \tilde{u}(x)}+\abs{\nabla \tilde{v}(x)}\leq C\frac{e^{-2(\abs{x}-1)}}{(\abs{x}-1)^{\frac{1}{2}}}\leq C\frac{e^{-2\abs{x}}}{\abs{x}^{\frac{1}{2}}},
\end{equation}
from which we obtain
\begin{equation}
\abs{\nabla u(x)}+\abs{\nabla v(x)}\leq C\frac{e^{-2\abs{x}}}{\abs{x}^{\frac{1}{2}}}.
\end{equation}
The proof of Theorem \ref{asym} is complete.

\end{proof}

\section{Solutions over a doubly-periodic domain}
In this section, we establish existence of solutions to (\ref{maineq}) for any $\abs{K}<0$ over a doubly-periodic domain, denoted by $\Omega$, known as the 2-torus.  In what follows, we will make the identification of the quotient space $\mathbb{R}^2/\mathbb{Z}^2$ with $\Omega$.  In other words, we may consider $\Omega$ to be a doubly-periodic grid in $\mathbb{R}^2$. With this domain in mind, and for some  numbers $\tau_1,\tau_2>0$, the system given by (\ref{maineq}) will have boundary conditions $u(x_1+\tau_1,x_2+\tau_2)=u(x_1,x_2),\quad v(x_1+\tau_1,x_2+\tau_2)=v(x_1,x_2).$

Again, we regularize our system.  As in  \cite{aubin1982nonlinear,medina}, we find two functions $u_0$ and $v_0$ which are unique up to an additive constant such that
\begin{equation}
\Delta u_0=-\frac{4\pi N_1}{\abs{\Omega}}+4\pi\sum\limits_{j=1}^{N_1}\delta_{p_j}\qquad\text{and}\qquad
\Delta v_0=-\frac{4\pi N_2}{\abs{\Omega}}+4\pi\sum\limits_{j=1}^{N_2}\delta_{q_j}.
\end{equation}
We let $u$ and $v$ be solutions to (\ref{maineq}) over $\Omega$ and write $u_1=u-u_0$, $v_1=v-v_0$.  Then, we transform (\ref{maineq}) into
\begin{equation}\label{doubleEQ1}
\begin{cases}
\Delta u_1=4k_{11}e^{u_0+u_1}+4k_{12}e^{v_0+v_1}-4+\frac{4\pi N_1}{\abs{\Omega}}&\text{in }\Omega\\
\Delta v_1=4k_{12}e^{u_0+u_1}+4k_{11}e^{v_0+v_1}-4+\frac{4\pi N_2}{\abs{\Omega}}&\text{in }\Omega\\
u_1(x_1+\tau_1,x_2+\tau_2)=u_1(x_1,x_2),\quad v_1(x_1+\tau_1,x_2+\tau_2)=v_1(x_1,x_2)
\end{cases}
\end{equation}

The work in the rest of this section is as follows.  First, we establish necessary conditions for existence of solutions to (\ref{doubleEQ1}).  Then, we transform the system so that we may establish a variational principle.  We will see that the functional obtained will be indefinite.  As a result, we will have to use the method of partial coercivity to obtain a critical point of the functional.
\subsection{Necessary conditions}
In this section, we will establish the necessary conditions for existence of solutions to (\ref{doubleEQ1}).  We first integrate both of the equations in (\ref{doubleEQ1}) over the doubly-periodic domain, $\Omega$  to obtain

\begin{equation}\label{nec1}
\begin{cases}
k_{11}\displaystyle\int\displaylimits_{\Omega}e^{u_0+u_1}dx+k_{12}\int\displaylimits_{\Omega}e^{v_0+v_1}dx=\abs{\Omega}-{\pi N_1}\\
k_{12}\displaystyle\int\displaylimits_{\Omega}e^{u_0+u_1}dx+k_{11}\int\displaylimits_{\Omega}e^{v_0+v_1}dx=\abs{\Omega}-{\pi N_2}.
\end{cases}
\end{equation}
Then, we use  $\abs{K}=4\frac{q}{p}$ and solve for $\int\displaylimits_{\Omega}e^{u_0+u_1}dx$ and $\int\displaylimits_{\Omega}e^{v_0+v_1}dx$ in (\ref{nec1}).  To this end we obtain
\begin{equation}\label{eu1}
\int\displaylimits_{\Omega}e^{u_0+u_1}dx=\frac{\abs{\Omega}}{2}-\frac{\pi}{4}\left[\left(N_1-N_2\right)\frac{p}{q}+\left(N_1+N_2\right)\right]
\end{equation}
and
\begin{equation}\label{ev1}
\int\displaylimits_{\Omega}e^{v_0+v_1}dx=\frac{\abs{\Omega}}{2}-\frac{\pi}{4}\left[\left(N_2-N_1\right)\frac{p}{q}+\left(N_1+N_2\right)\right].
\end{equation}
We let $N_-=N_1-N_2$ and $N_+=N_1+N_2$ and recall that since $\abs{K}<0$ and $p>0$ in the derivation of the system, we have that $q<0$.  We also see that the left hand sides of (\ref{eu1}) and (\ref{ev1}) are positive.  With this in mind, we are able to obtain the following necessary requirement for the solutions to (\ref{doubleEQ1}) to exist and the first part of Theorem \ref{necessary} has been established.

%

\subsection{Partial Coercivity and Minimization}
We let $\xi=u_1+v_1$ and $\zeta=u_1-v_1$, and we are able to rewrite the system given by (\ref{doubleEQ1}) as
\begin{equation}\label{doubleEQQ}
\begin{cases}
\Delta\xi=8e^{u_0+\frac{1}{2}(\xi+\zeta)}+8e^{v_0+\frac{1}{2}(\xi-\zeta)}-8+\frac{4\pi N_+}{\abs{\Omega}}\\
\Delta\zeta=2\abs{K}e^{u_0+\frac{1}{2}(\xi+\zeta)}-2\abs{K}e^{v_0+\frac{1}{2}(\xi-\zeta)}+\frac{4\pi N_-}{\abs{\Omega}}.
\end{cases}
\end{equation}
As we did when the system was considered over bounded domains, we multiply the first equation by ${\abs{K}}/{4}$ to obtain
\begin{equation}\label{doubleEQ2}
\begin{cases}
\frac{\abs{K}}{4}\Delta\xi=2\abs{K}e^{u_0+\frac{1}{2}(\xi+\zeta)}+2\abs{K}e^{v_0+\frac{1}{2}(\xi-\zeta)}-2\abs{K}+\frac{\pi\abs{K} N_+}{\abs{\Omega}}\\
\Delta\zeta=2\abs{K}e^{u_0+\frac{1}{2}(\xi+\zeta)}-2\abs{K}e^{v_0+\frac{1}{2}(\xi-\zeta)}+\frac{4\pi N_-}{\abs{\Omega}}.
\end{cases}
\end{equation}
We see that the equations in (\ref{doubleEQ2}) are the Euler-Lagrange equations of the functional
\begin{multline}\label{idouble}
I(\xi,\zeta)=\int\displaylimits_{\Omega}\bigg[\frac{\abs{K}}{8}\abs{\nabla\xi}^2+\frac{1}{2}\abs{\nabla\zeta}^2+4\abs{K}e^{u_0+\frac{1}{2}(\xi+\zeta)}+4\abs{K}e^{v_0+\frac{1}{2}(\xi-\zeta)}\\-2\abs{K}\xi+\frac{\pi\abs{K}N_+}{\Omega}\xi+\frac{4\pi N_-}{\abs{\Omega}}\zeta \bigg]dx.
\end{multline}
We then consider the following minimization problem,
\begin{equation}\label{doublemin}
\min\left\{I(\xi,\zeta)|(\xi,\zeta)\in\mathcal{A}\right\},
\end{equation}
where the admissible class, $\mathcal{A}$, is defined by
\begin{equation}\label{doubleadmis}
\mathcal{A}=\left\{(\xi,\zeta)|\xi,\zeta\in W^{1,2}(\Omega)\text{ and }\xi,\zeta\text{ satisfy (ii)}\right\}.
\end{equation}
\begin{align*}
\text{(ii)}\quad&\int\displaylimits_{\Omega}\bigg[\nabla\xi\cdot\nabla w+8we^{u_0+
\frac{1}{2}(\xi+\zeta)}+8we^{v_0+\frac{1}{2}(\xi-\zeta)}-8w+\frac{4\pi N_+}{\abs{\Omega}}w\bigg]dx=0,\\
&\forall w\in W^{1,2}(\Omega).
\end{align*}
Before establishing existence of a solution to the minimization problem (\ref{doublemin}), we observe that upon integration over the doubly-periodic domain, $\Omega$, the equations in (\ref{doubleEQ2}) yield the natural constraints
\begin{equation}\label{doublecon1}
\int\displaylimits_{\Omega}e^{u_0+\frac{1}{2}(\xi+\zeta)}dx=\frac{\abs{\Omega}}{2}-\frac{\pi}{4}\left(\frac{p}{q}N_-+N_+\right)=\alpha
\end{equation}
and
\begin{equation}\label{doublecon2}
\int\displaylimits_{\Omega}e^{v_0+\frac{1}{2}(\xi-\zeta)}dx=\frac{\abs{\Omega}}{2}-\frac{\pi}{4}\left(-\frac{p}{q}N_-+N_+\right)=\beta.
\end{equation}
In order for the constraints to make sense, we see that we must have $\alpha,\beta>0$.  That is, the size of the domain must satisfy, 
\begin{equation}
\abs{\Omega}>\frac{\pi}{2}\left(\abs{\frac{p}{q}}\abs{N_-}+N_+\right),
\end{equation}
and the sufficient condition in Lemma \ref{necessary} has been established.  

Now, we let $\mathscr{H}=W^{1,2}(\Omega)$ and decompose the space as $\mathscr{H}=\tilde{\mathscr{H}}\oplus\mathbb{R}$ where $\tilde{\mathscr{H}}$ is the closed subspace given by
\begin{equation}
\tilde{\mathscr{H}}=\left\{u\in\mathscr{H}\Big|\int\displaylimits_{\Omega}udx=0\right\}.
\end{equation}
With this decomposition, we are able to write $\xi=\tilde{\xi}+\bar{\xi}$ and $\zeta=\tilde{\zeta}+\bar{\zeta}$ where $\tilde{\xi},\tilde{\zeta}\in\tilde{\mathscr{H}}$ and $\bar{\xi},\bar{\zeta}\in\mathbb{R}$.  In this view, we rewrite (\ref{doublecon1}) and (\ref{doublecon2}) as
\begin{equation}\label{doublecon3}
\int\displaylimits_{\Omega}e^{u_0+\frac{1}{2}(\tilde{\xi}+\tilde{\zeta})+\frac{1}{2}\bar{\xi}+\frac{1}{2}\bar{\zeta}}dx=\alpha\qquad\text{and}\qquad
\int\displaylimits_{\Omega}e^{v_0+\frac{1}{2}(\tilde{\xi}-\tilde{\zeta})+\frac{1}{2}\bar{\xi}-\frac{1}{2}\bar{\zeta}}dx=\beta.
\end{equation}
From (\ref{doublecon3}), we see that
\begin{align}\label{doublecon4}
e^{\frac{1}{2}\bar{\xi}+\frac{1}{2}\bar{\zeta}}=\alpha\left(\int\displaylimits_{\Omega}e^{u_0+\frac{1}{2}(\tilde{\xi}+\tilde{\zeta})}dx \right)^{-1}\qquad\text{and}\qquad
e^{\frac{1}{2}\bar{\xi}-\frac{1}{2}\bar{\zeta}}=\beta\left(\int\displaylimits_{\Omega}e^{v_0+\frac{1}{2}(\tilde{\xi}-\tilde{\zeta})}dx \right)^{-1}.
\end{align}
Upon solving for $\bar{\xi}$ and $\bar{\zeta}$ in (\ref{doublecon4}), we obtain
\begin{align}
\begin{split}
\bar{\xi}&=\ln\left(\alpha\beta\right)-\ln\left( \int\displaylimits_{\Omega}e^{u_0+\frac{1}{2}(\tilde{\xi}+\tilde{\zeta})}dx\right) -\ln\left( \int\displaylimits_{\Omega}e^{v_0+\frac{1}{2}(\tilde{\xi}-\tilde{\zeta})}dx \right)\\
\bar{\zeta}&=\ln\left(\frac{\alpha}{\beta}\right)+\ln\left( \int\displaylimits_{\Omega}e^{u_0+\frac{1}{2}(\tilde{\xi}+\tilde{\zeta})}dx \right) -\ln\left(\int\displaylimits_{\Omega}e^{v_0+\frac{1}{2}(\tilde{\xi}-\tilde{\zeta})}dx \right).
\end{split}
\end{align}
We then use the Trudinger-Moser inequality, along with Jensen's inequality, to obtain the following estimates
\begin{equation}\label{alphabetabar}
-\ln\left(\alpha\beta\right)-C_1\leq \bar{\xi}\leq \ln\left(\alpha\beta\right)+C_2\qquad\text{and}\qquad
-\ln\left(\frac{\alpha}{\beta}\right)-C_1'\leq \bar{\zeta}\leq \ln\left(\frac{\alpha}{\beta}\right)+C_2'.
\end{equation}
Since the functional given by (\ref{idouble}) is indefinite, we will see it necessary to separate the functional into the two parts denoted below
\begin{equation}\label{doubleij}
I(\xi,\zeta)=\int\displaylimits_{\Omega}\bigg[\frac{1}{2}\abs{\nabla\zeta}^2+\frac{4\pi N_-}{\abs{\Omega}}\zeta \bigg]dx- J_{\zeta}(\xi),
\end{equation}
where
\begin{equation}\label{jzeta}
J_{\zeta}(\xi)=\int\displaylimits_{\Omega}\bigg[-\frac{\abs{K}}{8}\abs{\nabla\xi}^2-4\abs{K}e^{u_0+\frac{1}{2}(\xi+\zeta)}-4\abs{K}e^{v_0+\frac{1}{2}(\xi-\zeta)}+2\abs{K}\xi-\frac{\pi\abs{K}N_+}{\Omega}\xi\bigg] dx.
\end{equation}
We will now show that the functional $J_{\zeta}(\xi)$ is coercive, bounded below, and weakly lower semicontinuous.
\begin{lem}\label{jminlem}
Definition (ii) is well posed.  In other words, for any $\zeta\in W^{1,2}(\Omega)$, with $\Omega$ being a doubly-periodic domain, there exists a unique $\xi\in W^{1,2}(\Omega)$ satisfying (ii).  Furthermore, $\xi$ is the global minimizer of the functional $J_{\zeta}(\xi)$ given by (\ref{jzeta}).
\end{lem}

\begin{proof}
We first use the decomposition of $\xi=\tilde{\xi}+\bar{\xi}$ and $\zeta=\tilde{\zeta}+\bar{\zeta}$ in (\ref{jzeta}) to obtain
\begin{multline}\label{jzetamin1}
J_{\zeta}(\xi)=-\frac{\abs{K}}{8}\|\nabla\xi\|_{L^2(\Omega)}^2-4\abs{K}e^{\frac{\bar{\xi}}{2}}e^{\frac{\bar{\zeta}}{2}}\int\displaylimits_{\Omega}e^{u_0+\frac{1}{2}(\tilde{\xi}+\tilde{\zeta})}dx-4\abs{K}e^{\frac{\bar{\xi}}{2}}e^{\frac{\bar{\zeta}}{2}}\int\displaylimits_{\Omega}e^{v_0+\frac{1}{2}(\tilde{\xi}-\tilde{\zeta})}dx+2\abs{K}\int\displaylimits_{\Omega}\tilde{\xi} dx\\-\frac{\pi\abs{K}N_+}{\abs{\Omega}}\int\displaylimits_{\Omega}\tilde{\xi}dx+2\abs{K}\abs{\Omega}\bar{\xi}-\pi\abs{K}N_+\bar{\xi}.
\end{multline}
We now use the fact that $-4\abs{K}>0$ and the exponential terms in (\ref{jzetamin1}) are positive to obtain,
\begin{equation}\label{jzetamin2}
J_{\zeta}(\xi)\geq-\frac{\abs{K}}{8}\|\nabla\xi\|_{L^2(\Omega)}^2+\left(2\abs{K}-\frac{\pi\abs{K}N_+}{\abs{\Omega}}\right)\int\displaylimits_{\Omega}\tilde{\xi} dx+\left(2\abs{K}\abs{\Omega}-\pi\abs{K}N_+\right)\bar{\xi}.
\end{equation}
From the estimates obtained in (\ref{alphabetabar}), we know that the term $\bar{\xi}$ is bounded.  Therefore, we may further estimate the inequality in (\ref{jzetamin2}) to be
\begin{equation}\label{jzetamin3}
J_{\zeta}(\xi)\geq-\frac{\abs{K}}{8}\|\nabla\xi\|_{L^2(\Omega)}^2+\left(2\abs{K}-\frac{\pi\abs{K}N_+}{\abs{\Omega}}\right)\int\displaylimits_{\Omega}\tilde{\xi} dx-C(\bar{\xi},\abs{K},\abs{\Omega},N_+),
\end{equation}
where $C(\bar{\xi},\abs{K},\abs{\Omega},N_+)>0.$  Finally, we recall that $\tilde{\xi}$ satisfies $\int\displaylimits_{\Omega}\tilde{\xi}dx=0$, so that from (\ref{jzetamin3}) we obtain
\begin{equation}\label{jzetamin4}
J_{\zeta}(\xi)\geq-\frac{\abs{K}}{8}\|\nabla\xi\|_{L^2(\Omega)}^2-C(\bar{\xi},\abs{K},\abs{\Omega},N_+).
\end{equation}
Since $\abs{K}<0$, we see that the coefficient of $\|\nabla\xi\|$ is positive.  Therefore, we have established that $J_{\zeta}(\xi)$ is bounded below and coercive. We also note that it is easy to see that $J_{\zeta}(\xi)$ is weakly lower semicontinuous.

Therefore, for any $\zeta\in W^{1,2}(\Omega)$, we find a minimizer, $\xi$ of $J_{\zeta}(\cdot)$.  By the convexity of $J_{\zeta}(\cdot)$ in $\xi$, we see that the minimizer, $\xi$, is unique.  Moreover, since the first variation of the functional (\ref{jzeta}) is actually the left-hand side condition (ii), we see that condition (ii) is valid.
\end{proof}

\begin{thm}
The functional, $I(\xi,\zeta)$ given by (\ref{idouble}) is partially coercive.
\end{thm}
\begin{proof}
First, we note that since for any $\zeta$ the functional $J_{\zeta}(\xi)$ is minimized uniquely by $\xi$, we have
\begin{equation}\label{jzetazero}
 J_{\zeta}(\xi)\leq J_{\zeta}(0)=-4\abs{K}\int\displaylimits_{\Omega}\bigg[e^{u_0+\frac{1}{2}\zeta}+e^{v_0-\frac{1}{2}\zeta}\bigg]dx
\leq C\int\displaylimits_{\Omega}\bigg[e^{\frac{1}{2}\tilde{\zeta}}+e^{-\frac{1}{2}\tilde{\zeta}}\bigg]dx
\leq2C\int\displaylimits_{\Omega}e^{\frac{1}{4}\tilde{\zeta}^2}dx\\
\leq C'.
\end{equation}
Here we have used the Trudinger-Moser inequality along with the fact that $\int\displaylimits_{\Omega}\tilde{\zeta}dx=0$.  This allows us to obtain
\begin{align}\label{doublecoercive}
\nonumber I(\xi,\zeta)&=\int\displaylimits_{\Omega}\frac{1}{2}\abs{\nabla\zeta}^2+\frac{4\pi N_-}{\abs{\Omega}}\zeta dx- J_{\zeta}(\xi)\\
&\nonumber\geq\frac{1}{2}\|\nabla\zeta\|_{L^2(\Omega)}^2+4\pi N_-\bar{\zeta}-J_{\zeta}(0)\\
&\geq\frac{1}{2}\|\nabla\zeta\|_{L^2(\Omega)}^2 -C(N_-,\alpha,\beta).
\end{align}
We now see that the functional given by (\ref{idouble}) is partially coercive and bounded below.  Furthermore, it is easy to see that the functional (\ref{idouble}) is weakly lower semicontinuous.   Therefore, a minimizer exists.

%

\end{proof}

\begin{lem}\label{weakdoublemin}
A minimizing sequence $\left\{(\xi_n,\zeta_n)\right\}_{n=1}^{\infty}$ of (\ref{doublemin}) satisfies $\xi_n\rightharpoonup\xi$ weakly in $W^{1,2}(\Omega)$ and $\zeta_n\rightharpoonup\zeta$ weakly in $W^{1,2}(\Omega)$ as $n\rightarrow\infty$.
\end{lem}

\begin{proof}
We let $\left\{(\xi_n,\zeta_n)\right\}_{n=1}^{\infty}$ be any minimizing sequence of (\ref{doublemin}) such that
\begin{equation*}
I(\xi_1,\zeta_1)\geq I(\xi_2,\zeta_2)\geq\cdots\geq I(\xi_n,\zeta_n)\geq\cdots.
\end{equation*}
We then define
\begin{equation}
\eta_0:=\inf\left\{I(\xi,\zeta)\mid (\xi,\zeta)\in\mathcal{A}\right\}=\lim\limits_{n\rightarrow\infty}I(\xi_n,\zeta_n).
\end{equation}
It is easy to see that by (\ref{doublecoercive}), the sequence $\left\{\zeta_n\right\}_{n=1}^{\infty}$ satisfies
\begin{equation}\label{doubbound3}
M_0\geq I(\xi_n,\zeta_n)\geq\frac{1}{2}\|\nabla\zeta_n\|_{L^2(\Omega)}^2-C(N_-,\alpha,\beta).
\end{equation}
where the upper bound, $M_0$, comes from the fact that a minimizing sequence is bounded.  By (\ref{doubbound3}), we see that $\nabla\zeta_n$ is bounded in $L^2(\Omega)$.  Furthermore, the Poincar\'e inequality on $W^{1,2}(\Omega)$ with $\int\displaylimits_{\Omega}fdx=0$ gives
\begin{equation}\label{dpoin1}
\|\zeta_n\|_{L^2(\Omega)}^2\leq C\|\nabla\zeta_n\|_{L^2(\Omega)}^2.
\end{equation}

When we insert (\ref{dpoin1}) into (\ref{doublecoercive}), we establish that $\left\{\zeta_n\right\}_{n=1}^{\infty}$ is bounded in $L^2(\Omega)$.  Notice that we get
\begin{equation}\label{doublein46}
M_0\geq C_2\|\zeta_n\|_{L^2(\Omega)}^2-C(N_-,\alpha,\beta).
\end{equation}
for some $C_2>0$.  Inequality (\ref{doublein46}) implies that $\left\{\zeta_n\right\}_{n=1}^{\infty}$ is bounded in $L^2(\Omega)$, and ultimately bounded in  $W^{1,2}(\Omega)$. 

 In a similar fashion, we use (\ref{jzetazero}) to obtain
\begin{equation}\label{doublein47}
M_0'\geq-\frac{\abs{K}}{8}\|\nabla\xi_n\|_{L^2(\Omega)}^2-C(\bar{\xi_n},\abs{K},\abs{\Omega},N_+)\geq C_3\|\xi_n\|_{L^2(\Omega)}^2-C(\bar{\xi_n},\abs{K},\abs{\Omega},N_+)
\end{equation}
where $C_3>0$.  Inequality (\ref{doublein47}) establishes the boundedness of $\left\{\xi_n\right\}_{n=1}^{\infty}$ in $W^{1,2}(\Omega)$.  We then conclude that (passing to a subsequence if necessary) $\xi_n\rightharpoonup \xi$ weakly in $W^{1,2}(\Omega)$.
\end{proof}
We now show that the pair of functions, $(\xi,\zeta)$ found in Lemma \ref{weakdoublemin} belong to the admissible class, $\mathcal{A}$, given by (\ref{doubleadmis}).

\begin{lem}\label{doubleadmislem}
The functions $\xi$ and $\zeta$ defined in Lemma \ref{weakdoublemin} satisfy $(\xi,\zeta)\in\mathcal{A}$
\end{lem}

\begin{proof}
Since $\xi_n,\zeta_n\in\mathcal{A}$, we know that they satisfy condition (ii).  We can write this as
\begin{multline}\label{double18}
\int\displaylimits_{\Omega}\bigg[\nabla \xi_n\cdot\nabla w+8we^{u_0+\frac{1}{2}(\xi_n+\zeta_n)}+8we^{v_0+\frac{1}{2}(\xi_n-\zeta_n)}-8w+\frac{4\pi N_+}{\abs{\Omega}}w\bigg]dx=0,\qquad\forall w\in W^{1,2}(\Omega).
\end{multline}
Now, we recall that $\xi_n\rightharpoonup\xi$ and $\zeta_n\rightharpoonup\zeta$ weakly in $W^{1,2}(\Omega)$.  From this weak convergence, we know that we can find subsequences, denoted by $\left\{\xi_{n_j}\right\}_{j=1}^{\infty}$ and $\left\{\zeta_{n_j}\right\}_{j=1}^{\infty}$, such that they converge strongly to $\xi$ and $\zeta$ in $L^2(\Omega)$.  Without loss of generality (passing to a subsequence if necessary), we will assume that $\xi_n\rightarrow\xi$, and $\zeta_n\rightarrow\zeta$ strongly in $L^2(\Omega)$ as $n\rightarrow\infty$.  As in section 3, we obtain convergence of the exponential integrals by way of the Trudinger-Moser inequality, the general H\"older inequality and the mean value theorem.  Details are omitted here as they are a repetition of the work in section 3.  Therefore, we have
\begin{equation}\label{doublee1}
\int\displaylimits_{\Omega}e^{u_0+\frac{1}{2}(\xi_n+\zeta_n)}dx\rightarrow\int\displaylimits_{\Omega}e^{u_0+\frac{1}{2}(\xi+\zeta)}dx,
\end{equation}
and
\begin{equation}\label{doublee2}
\int\displaylimits_{\Omega}e^{b_0+\frac{1}{2}(\xi_n-\zeta_n)}dx\rightarrow\int\displaylimits_{\Omega}e^{v_0+\frac{1}{2}(\xi-\zeta)}dx.
\end{equation}
%
We now take $n\rightarrow\infty$ in (\ref{double18}) and see that $(\xi,\zeta)\in\mathcal{A}$. 
\end{proof}

\begin{lem}\label{doublestrong}
The convergence of $\xi_n\rightarrow\xi$ is strong in $W^{1,2}(\Omega)$
\end{lem}
\begin{proof}
First, we let $w=\xi_n-\xi$ in the equations given by (ii) and (\ref{double18}).  With this substitution, we end up with both
\begin{equation}\label{doublez1}
\int\displaylimits_{\Omega}\bigg[\nabla\xi\cdot\nabla (\xi_n-\xi)+8(\xi_n-\xi)\left(e^{u_0+\frac{1}{2}(\xi+\zeta)}+e^{v_0+\frac{1}{2}(\xi-\zeta)}-1\right)+\frac{4\pi N_+}{\abs{\Omega}}(\xi_n-\xi)\bigg]dx=0
\end{equation}
and
\begin{multline}\label{doublez2}
\int\displaylimits_{\Omega}\bigg[\nabla \xi_n\cdot\nabla (\xi_n-\xi)+8(\xi_n-\xi)\left(e^{u_0+\frac{1}{2}(\xi_n+\zeta_n)}+e^{v_0+\frac{1}{2}(\xi_n-\zeta_n)}-1\right)+\frac{4\pi N_+}{\abs{\Omega}}(\xi_n-\xi)\bigg]dx=0.
\end{multline}
Next, we subtract (\ref{doublez2}) from (\ref{doublez1}) to obtain,
\begin{equation}\label{doublez3}
\int\displaylimits_{\Omega}\bigg[-\abs{\nabla\xi_n-\nabla\xi}^2+8(\xi_n-\xi)\left\{e^{u_0}\left(e^{\frac{1}{2}(\xi+\zeta)}-e^{\frac{1}{2}(\xi_n+\zeta_n)}\right)+e^{v_0}\left(e^{\frac{1}{2}(\xi-\zeta)}-e^{\frac{1}{2}(\xi_n-\zeta_n)} \right)\right\}\bigg]dx=0.
\end{equation}
We are able to rewrite (\ref{doublez3}) and take the limit as $n\rightarrow\infty$ to observe that,
\begin{multline}\label{doublez4}
\int\displaylimits_{\Omega}\abs{\nabla\xi_n-\nabla\xi}^2dx=\int\displaylimits_{\Omega}\bigg[8(\xi_n-\xi)e^{u_0}\left(e^{\frac{1}{2}(\xi+\zeta)}-e^{\frac{1}{2}(\xi_n+\zeta_n)}\right)\\+8(\xi_n-\xi)e^{v_0}\left(e^{\frac{1}{2}(\xi-\zeta)}-e^{\frac{1}{2}(\xi_n-\zeta_n)} \right)\bigg]dx\rightarrow 0\text{ as }n\rightarrow\infty.
\end{multline}
Therefore, we have shown that $\nabla\xi_n\rightarrow\nabla\xi$ in $L^2(\Omega)$.  We may conclude now that $\xi_n\rightarrow\xi$ strongly in $W^{1,2}(\Omega)$.
\end{proof}

We now state the main theorem of this section,
\begin{thm}\label{doubleminsolution}
There exists a pair $(\xi,\zeta)\in\mathcal{A}$ that is a solution to the minimization problem (\ref{doublemin}).
\end{thm}

\begin{proof}
We recall that $\zeta_n\rightharpoonup\zeta$ weakly in $W^{1,2}(\Omega)$ and $\xi_n\rightarrow\xi$ strongly in $W^{1,2}(\Omega).$ As a result, we have
\begin{equation}\label{double419}
\lim\limits_{n\rightarrow\infty}J_{\zeta_n}(\xi_n)=J_{\zeta}(\xi).
\end{equation}
We use (\ref{doubleij}), (\ref{double419}), and recall that $\displaystyle\eta_0:=\inf\left\{I(\xi,\zeta)\mid (\xi,\zeta)\in\mathcal{A}\right\}=\lim\limits_{n\rightarrow\infty}I(\xi_n,\zeta_n)$ to obtain the estimate
\begin{equation}
\eta_0=\lim\limits_{n\rightarrow\infty}I(\xi_n,\zeta_n)\geq \frac{1}{2}\|\nabla\zeta\|_{L^2(\Omega)}^2+\int\displaylimits_{\Omega}\frac{4\pi N_-}{\abs{\Omega}}\zeta dx-J_{\zeta}(\xi)\geq I(\xi,\zeta)\geq\eta_0.
\end{equation}
Since we have already established that $(\xi,\zeta)\in\mathcal{A}$, we see that the pair solves the minimization problem given by (\ref{doublemin}).
\end{proof}

The final lemma of this section, establishing existence of a solution to the system (\ref{doubleEQ2}), is stated below.
\begin{lem}\label{doublesolution}
The pair defined in Lemma \ref{weakdoublemin} is a solution to the system given by (\ref{doubleEQ2}) and ultimately to (\ref{doubleEQ1}).
\end{lem}
\begin{proof}
First, we point out that the weak form of the second equation in (\ref{doubleEQQ}) is indeed the condition (ii) in (\ref{doubleadmis}).  Therefore, we only need to verify the first equation in (\ref{doubleEQQ}).   We proceed in a similar manner to that in Lemma \ref{chainrule}.  That is, we let $\hat{\zeta}$ be any test function in $W^{1,2}(\Omega)$ and set $\zeta_{\tau}=\zeta+\tau\hat{\zeta}$.  Since we have shown that for any $\zeta$ there is a unique minimizer $\xi$ of $J_{\zeta}(\xi)$, we denote the unique minimizer of $J_{\zeta_{\tau}}(\cdot)$ by $\xi_{\tau}$ and again let 
\begin{equation}
\tilde{\xi}=\left(\frac{d}{d\tau}\xi_{\tau}\right)\Bigg|_{\tau=0}.
\end{equation}
where $\xi_{\tau}$ depends upon $\zeta_{\tau}$.  Since the pair $(\xi,\zeta)$ minimizes $I(\cdot,\cdot)$, we see that $I(\xi_{\tau},\zeta_{\tau})$ attains its minimum at $\tau=0$.  Therefore, we have
\begin{equation}\label{doubleid}
\frac{d}{d\tau}I(\xi_{\tau},\zeta_{\tau})\Big|_{\tau=0}=0.
\end{equation}
Then, in view of the formulation of $I(\xi,\zeta)$ given by (\ref{idouble}), we rewrite (\ref{doubleid}) as
\begin{align}\label{double67}
&\int\displaylimits_{\Omega}\bigg[\nabla\zeta\cdot\nabla\hat{\zeta}+2\abs{K}e^{u_0+\frac{1}{2}(\xi+\zeta)}\hat{\zeta}-2\abs{K}e^{v_0+\frac{1}{2}(\xi-\zeta)}\hat{\zeta}+\frac{4\pi N_-}{\abs{\Omega}}\hat{\zeta}\bigg] dx\\&=\int\displaylimits_{\Omega}\bigg[-\frac{\abs{K}}{4}\nabla\xi\cdot\nabla\hat{\xi}-2\abs{K}e^{u_0+\frac{1}{2}(\xi+\zeta)}\hat{\xi}-2\abs{K}e^{v_0+\frac{1}{2}(\xi-\zeta)}\hat{\xi}+2\abs{K}\hat{\xi}-\frac{\pi\abs{K}N_+}{\Omega}\hat{\xi}\bigg]dx.
\end{align}
By the condition (ii) in (\ref{doubleadmis}), we see that the right hand side of (\ref{double67}) is zero.  Therefore, we have obtained
\begin{equation}
\int\displaylimits_{\Omega}\bigg[\nabla\zeta\cdot\nabla\hat{\zeta}+2\abs{K}e^{u_0+\frac{1}{2}(\xi+\zeta)}\hat{\zeta}-2\abs{K}e^{v_0+\frac{1}{2}(\xi-\zeta)}\hat{\zeta}+\frac{4\pi N_-}{\abs{\Omega}}\hat{\zeta}\bigg] dx=0
\end{equation}
for an arbitrary test function $\hat{\zeta}\in W^{1,2}(\Omega)$.  It is clear that this is the weak formulation of the first equation in (\ref{doubleEQQ}) and we have fully verified the system (\ref{doubleEQQ}).  
\end{proof}

\section{Discussion}
In this paper, we studied the existence of solutions to the vortex equations governing the fractional quantum Hall effect in two-dimensional bilayered systems.  We were able to extend the work done by Medina \cite{medina}, where solutions were established over the full-plane, as well as a doubly-periodic domain, for a positive definite coupling matrix, $K$.    

In this work, we established existence of topological solutions to the system (\ref{maineq}) over the full-plane when the coupling matrix is indefinite and satisfies $-4\leq\abs{K}<0$.  For these topological solutions, which are necessarily of finite energy, we also obtained that as $\abs{x}\rightarrow\infty$, the solutions decay as ${e^{-2\abs{x}}}/{\abs{x}^{1/2}}.$

We remark that when $\abs{K}<-4$, we were able to establish existence over a bounded domain, but were unable to extend the existence over a bounded domain to that of a full-plane.  

In addition to existence of solutions over the full-plane, we were also able to obtain solutions when the domain was a torus, or a doubly-periodic domain.  We established the necessary and sufficient conditions for solutions to occur.  That is, the domain, $\Omega$, must satisfy
\begin{equation*}
\abs{\Omega}>\frac{\pi}{2}\left(\abs{\frac{p}{q}}\abs{N_-}+N_+\right),
\end{equation*}
which agrees with the work done by Medina in \cite{medina}.  This work completes the existence theory for solutions to the system given by (\ref{maineq}) over doubly-periodic domains.  When the work of Medina \cite{medina} is combined with this paper, we see that the full range of $\abs{K}$ is covered for doubly-periodic domains.  Of future interest would be to establish existence theory of non-topological solutions to the system given by (\ref{maineq}) over the full-plane. 

Since we worked with an indefinite matrix, $K$, we had the parameter $q<0$.   We recall that in section 2, we discussed the impact of this parameter on the charge, pseudospin, and Chern-Simons flux.  We also noted that the filling factor, $\nu=\pi/p$ is independent of this filling factor.  Therefore, we may consider here a special filling factor, the much studied $\nu=5/2$ \cite{pfaf4,pfaf3,pfaf1}.  Here, we assume that $-4\leq\abs{K}<0$ so that solutions exist over both the full-plane and doubly periodic domains.  It is easy to see that since $\abs{K}=4q/p$, the parameter $q$ must satisfy $-1\leq q<0$.  Furthermore, we obtain the following for the charges of the upper and lower layers respectively
\begin{equation}
Q=-\frac{\pi}{2}N_+\qquad\text{and}\qquad \tilde{Q}=-\frac{\pi^2}{5q}N_-.
\end{equation}  
We see that the charge $\tilde{Q}$ will vary depending upon $q$, we recall that $N_-=N_1-N_2$, and we make the following observations:
\begin{enumerate}
\item If $N_->0$, then $\tilde{Q}\geq \frac{\pi^2}{5}N_-$.
\item If $N_-=0$, then $\tilde{Q}=0$.
\item If $N_-<0$, then $\tilde{Q}<\frac{\pi^2}{5}N_-$.
\end{enumerate} 
Therefore, for a filling factor of $\nu=5/2$, solutions to (\ref{maineq}) exist over the full-plane for all pairs $(p,q)=\left(\frac{2\pi}{5},q\right)$ where $-1\leq q<0$.  

Regarding the doubly periodic domain, we see that for $\nu=5/2$, the size of the domain must satisfy
\begin{equation}
\abs{\Omega}>\frac{\pi}{2}\left(\frac{2\pi}{5\abs{q}}\abs{N_-}+N_+\right).
\end{equation}

A similar analysis could be carried out for a variety of other possible fractional filling factors such as $\nu=4/5, 5/7, 6/7$, or even $\nu=1/(2k+1)$ \cite{nonconv,bondsling,cw,dasilva2016fractional,PhysRevB.95.125302,bondsling3}.

\bibliography{Indefinite}{}
\bibliographystyle{abbrv}

\end{document}